\documentclass[letterpaper,11pt]{article}

\usepackage[margin=1.1in]{geometry}
\usepackage{amsmath,amsfonts,latexsym,amssymb,xspace}
\usepackage{amsthm}
\usepackage{enumitem}
\usepackage{algorithm}
\usepackage[noend]{algorithmic}
\usepackage{engord}
\usepackage{cite}

\newtheorem{theorem}{Theorem}
\newtheorem{corollary}[theorem]{Corollary}
\newtheorem{lemma}[theorem]{Lemma}

\theoremstyle{remark}
\newtheorem{remark}{Remark}

\newcommand{\nats}{\ensuremath{\mathbb{N}}\xspace}
\newcommand{\ints}{\ensuremath{\mathbb{Z}}\xspace}
\newcommand{\rats}{\ensuremath{\mathbb{Q}}\xspace}

\newcommand{\nnrats}{\ensuremath{\rats^+}\xspace}

\newcommand{\csp}{\textsf{\textup{CSP}}\xspace}
\newcommand{\cspg}[1][\Gamma]{\textsf{\textup{CSP($#1$)}}\xspace}
\newcommand{\numcsp}{\textsf{\textup{\#CSP}}\xspace}
\newcommand{\ncsp}[1][\Gamma]{\textsf{\textup{\#CSP($#1$)}}\xspace}

\newcommand{\np}{\textsf{\textup{NP}}\xspace}
\newcommand{\npc}{\np{}-complete\xspace}
\newcommand{\nump}{\textsf{\textup{\#P}}\xspace}
\newcommand{\numpc}{\textsf{\textup{\#P}}-complete\xspace}
\newcommand{\numph}{\textsf{\textup{\#P}}-hard\xspace}
\newcommand{\ptime}{\textsf{\textup{P}}\xspace}

\newcommand{\fp}{\textsf{\textup{FP}}\xspace}

\newcommand{\PP}{\textsf{PP}\xspace}
\newcommand{\parityP}{\raisebox{1pt}{$\boldsymbol\oplus$}\textsf{P}\xspace}

\newcommand{\inj}{\hookrightarrow}
\newcommand{\bij}{\leftrightarrow}
\newcommand{\set}[1]{\{#1\}}
\newcommand{\size}[1]{\|#1\|}
\newcommand{\maltsev}{\textrm{Mal'tsev}\xspace}
\newcommand{\cocl}[1][\Gamma]{\ensuremath{\langle #1\rangle}\xspace}
\newcommand{\cl}[1]{\ensuremath{\textsf{\textup{cl}}_{\varphi} #1}\xspace}
\newcommand{\closure}{\textsc{Closure}\xspace}
\newcommand{\proj}[2][i]{\ensuremath{\mathsf{pr}_{#1} #2}\xspace}
\newcommand{\tw}[1][i]{\ensuremath{\sim_{#1}}\xspace}

\newcommand{\ph}{\ensuremath{\varphi}\xspace}
\newcommand{\equ}{\ensuremath{\boldsymbol{=}}\xspace}
\newcommand{\bigO}{\mathcal{O}}
\newcommand{\con}{\ensuremath{\Theta}\xspace}
\newcommand{\Hom}{\mbox{\textup{Hom}}}
\newcommand{\Mon}{\mbox{\textup{Mon}}}
\newcommand{\mon}{\mbox{\textup{mon}}}
\newcommand{\strbal}{\textsc{Strong Balance}\xspace}
\newcommand{\strect}{\textsc{Strong Rectangularity}\xspace}
\newcommand{\congsing}{\textsc{Congruence Singularity}\xspace}

\newcommand{\cA}{\ensuremath{\mathcal{A}}\xspace}

\newcommand{\cD}{\ensuremath{\mathcal{D}}\xspace}
\newcommand{\cE}{\ensuremath{\mathcal{E}}\xspace}
\newcommand{\cF}{\ensuremath{\mathcal{F}}\xspace}
\newcommand{\cG}{\ensuremath{\mathcal{G}}\xspace}

\newcommand{\cI}{\ensuremath{\mathcal{I}}\xspace}
\newcommand{\cM}{\ensuremath{\mathcal{M}}\xspace}
\newcommand{\cS}{\ensuremath{\mathcal{S}}\xspace}

\newcommand{\sW}{\ensuremath{\mathsf{W}}\xspace}
\newcommand{\ba}{\ensuremath{\mathbf{a}}\xspace}
\newcommand{\bb}{\ensuremath{\mathbf{b}}\xspace}
\newcommand{\bc}{\ensuremath{\mathbf{c}}\xspace}
\newcommand{\bd}{\ensuremath{\mathbf{d}}\xspace}
\newcommand{\Bf}{\ensuremath{\mathbf{f}}\xspace}
\newcommand{\bg}{\ensuremath{\mathbf{g}}\xspace}
\newcommand{\bh}{\ensuremath{\mathbf{h}}\xspace}
\newcommand{\bi}{\ensuremath{\mathbf{i}}\xspace}
\newcommand{\bm}{\ensuremath{\mathbf{m}}\xspace}

\newcommand{\bt}{\ensuremath{\mathbf{t}}\xspace}
\newcommand{\bu}{\ensuremath{\mathbf{u}}\xspace}
\newcommand{\bv}{\ensuremath{\mathbf{v}}\xspace}
\newcommand{\bw}{\ensuremath{\mathbf{w}}\xspace}
\newcommand{\bx}{\ensuremath{\mathbf{x}}\xspace}
\newcommand{\by}{\ensuremath{\mathbf{y}}\xspace}
\newcommand{\bz}{\ensuremath{\mathbf{z}}\xspace}
\newcommand{\bM}{\ensuremath{\mathbf{M}}\xspace}
\newcommand{\bP}{\ensuremath{\mathbb{P}}\xspace}

\newcommand{\euGam}{{\bar{\Gamma}}}
\newcommand{\euPhi}{{\bar{\Phi}}}
\newcommand{\euPsi}{{\bar{\Psi}}}
\newcommand{\eua}{{\bar{a}}}
\newcommand{\euc}{{\bar{c}}}
\newcommand{\eud}{{\bar{d}}}
\newcommand{\eui}{{i}}
\newcommand{\euq}{{\bar{q}}}
\newcommand{\euu}{{\bar{u}}}
\newcommand{\eus}{{\bar{s}}}
\newcommand{\euD}{{\bar{D}}}
\newcommand{\euH}{{\bar{H}}}
\newcommand{\euI}{{\bar{I}}}
\newcommand{\euM}{{\bar{M}}}
\newcommand{\euS}{{\ensuremath{\bar{\mathfrak{S}}}}\xspace}
\newcommand{\frS}{{\ensuremath{\mathfrak{S}}}\xspace}

\newcommand{\wit}{\ensuremath{\boldsymbol{\omega}}\xspace}

\newcommand{\phdash}{\phantom{{}'}} 

\floatstyle{plain}
\restylefloat{algorithm}

\title {An Effective Dichotomy for the\\ 
        Counting Constraint Satisfaction Problem}
\author{Martin Dyer%
\thanks{School of Computing, University of Leeds, Leeds, LS2 9JT, UK.}
        \and
        David Richerby%
\thanks{Department of Computer Science, University of Liverpool,
        Liverpool, L69 3BX, UK.\newline  This research was supported by EPSRC
        grants EP/E062172/1 ``The Complexity of Counting in Constraint
        Satisfaction Problems'' and EP/I012087/1  "Computational Counting".}}
\date{}

\begin{document}
\allowdisplaybreaks{}
\maketitle{}

\begin{abstract}
    \noindent Bulatov (2008) gave a dichotomy for the counting
    constraint satisfaction problem \numcsp. A problem from \numcsp is
    characterised by a constraint language $\Gamma\!$, a fixed, finite
    set of relations over a finite domain $D$. An instance of the
    problem uses these relations to constrain an arbitrarily large
    finite set of variables.  Bulatov showed that the problem of
    counting the satisfying assignments of instances of any problem
    from \numcsp is either in polynomial time (\fp) or is \numpc. His
    proof draws heavily on techniques from universal algebra and
    cannot be understood without a secure grasp of that field. We give
    an elementary proof of Bulatov's dichotomy, based on succinct
    representations, which we call \emph{frames}, of a class of highly
    structured relations, which we call \emph{strongly
      rectangular}. We show that these are precisely the relations
    which are invariant under a \emph{\maltsev polymorphism}. En
    route, we give a simplification of a decision algorithm for
    strongly rectangular constraint languages, due to Bulatov and
    Dalmau (2006). We establish a new criterion for the \numcsp
    dichotomy, which we call \emph{strong balance}, and we prove that
    this property is decidable. In fact, we establish membership in
    \np. Thus, we show that the dichotomy is effective,
    resolving the most important open question concerning the \numcsp
    dichotomy.
\end{abstract}


\section{Introduction}

The constraint satisfaction problem (\csp) is ubiquitous in computer
science.  Problems in such diverse areas as Boolean logic, graph
theory, database query evaluation, type inference, scheduling and
artificial intelligence can be expressed naturally in the setting of
assigning values from some domain to a collection of variables,
subject to constraints on the combinations of values taken
by given tuples of variables~\cite{FedVar98}.  \csp is directly
equivalent to the problem of evaluating conjunctive queries on
databases~\cite{KolVar98a} and to the homomorphism problem for
relational structures~\cite{FedVar98}.  Weighted versions of \csp
appear in statistical physics, where the total weight of solutions
corresponds to the so-called partition function of a spin
system~\cite{DyeGre00}.

For example, suppose we wish to know if a graph is 3-colourable.  The
question we are trying to answer is whether we can assign a colour
(domain value) to each vertex (variable) such that, whenever two
vertices are adjacent in the graph, they receive a different colour
(constraints).  Similarly, by asking if a 3-CNF formula is
satisfiable, we are asking if we can assign a truth value to each
variable such that every clause contains at least one true literal.

Since it includes both 3-\textsc{colourability} and 3-\textsc{sat},
this general form of the \csp, known as \emph{uniform} \csp, is \npc{}.
Therefore, attention has focused on \emph{nonuniform} \csp.  Here, we
fix a domain and a finite \emph{constraint language} $\Gamma\!$, a set
of relations over that domain.  Having fixed $\Gamma\!$, we only allow
constraints of the form, ``the values assigned to the variables $v_1,
\dots, v_r$ must be a tuple in the $r$-ary relation $R\in\Gamma$'' (we
define these terms formally in Section~\ref{sec:Defs}).  We write
$\csp(\Gamma)$ to denote nonuniform \csp with constraint language
$\Gamma\!$.  To express 3-\textsc{colourability} in this setting, we
just take $\Gamma$ to be the disequality relation on a set of three
colours.  3-\textsc{sat} is also expressible: to see this, observe
that, for example, the clause $\neg x \vee y \vee \neg z$ corresponds
to the relation $\{\texttt{t}, \texttt{f}\}^3 \setminus
\{\texttt{t}, \texttt{f}, \texttt{t}\}$, where $\texttt{t}$ indicates ``true'' and $\texttt{f}$ ``false'', and that the other seven patterns of negations within a clause can be expressed similarly.

Thus, there are languages $\Gamma$ for which $\csp(\Gamma)$ is \npc{}.
Of course, we can also express polynomial-time problems such as
2-\textsc{Colourability} and 2-\textsc{Sat}.  Feder and Vardi~\cite{FedVar98} conjectured that these are the only possibilities: that is, for all $\Gamma\!$, $\csp(\Gamma)$ is in \ptime{} or
is \npc{}.  To date, this conjecture remains open but it is known to hold in special cases~\cite{Schaef78, Bulato06, HelNes90}.  Recent
efforts to resolve the conjecture have focused on techniques from
universal algebra~\cite{DenWis02}.

There can be no dichotomy for the whole of \np{}, since
Ladner~\cite{Ladner75} has shown that either $\ptime=\np$ or there is
an infinite hierarchy of complexity classes between them. Hence,
assuming that $\ptime\neq\np$, there exist problems in \np{} that are
neither complete for the class nor in \ptime{}.  However, it is not
unreasonable to conjecture a dichotomy for \csp, since there are \np{}
problems, such as graph Hamiltonicity and even connectivity, that
cannot be expressed as $\csp(\Gamma)$ for any finite $\Gamma$.  This
follows from the observation that any set $S$ of structures (e.g.,
graphs) that is definable in \csp{} has the property that, if $A\in S$
and there is a homomorphism $B\to A$, then $B\in S$; neither the set
of Hamiltonian nor connected graphs has this property.  Further,
Ladner's theorem is proven by a diagonalisation that does not seem to
be expressible in \csp~\cite{FedVar98}.

In this paper, we consider the \emph{counting} version of $\csp(\Gamma)$,
which we denote \ncsp{}.  Rather than ask whether an instance of
$\csp(\Gamma)$ has a satisfying assignment, we ask how many satisfying assignments there are.   The corresponding conjecture was that, for every $\Gamma\!$, \ncsp{} is either computable in polynomial time or
complete for \nump{}.  We give formal
definitions in the next section but, informally, \nump{} is the analogue of \np{} for counting problems. Again, a modification of Ladner's proof shows that there can be no dichotomy for the whole of \nump. Note that the decision version of any problem in \np is trivially reducible to the corresponding counting problem in \nump: if we can count the number of solutions, we can certainly determine whether one exists.  However, the converse cannot hold under standard assumptions about complexity theory: there are well-known polynomial-time algorithms that determine whether a graph admits a perfect matching but it is \numpc{} to count the perfect matchings of even a bipartite
graph~\cite{Valian79a}.

Dichotomies for \ncsp{} are known in several special
cases~\cite{CreHer96, DyeGre00, DyGoPa07, DyGoJe08,CaLuXi09}, each consistent with the conjecture that \ncsp{} is always either polynomial-time computable or \numpc{}.  However, Bulatov recently made a major breakthrough by proving a dichotomy for all $\Gamma$~\cite{Bulato07,Bulato08}.

Bulatov's proof makes heavy use of the techniques of universal
algebra.  A relation is said to be \emph{pp-definable} over a
constraint language $\Gamma$ if it can be defined from the relations
in $\Gamma$ by a logical formula that uses only conjunction and
existential quantification.  Geiger~\cite{Geiger68} showed that an
algebra can be associated with the set of pp-definable relations over
$\Gamma$ and Bulatov examines detailed properties of the
\emph{congruence lattice} of this algebra.\footnote{We will not define
  these terms from universal algebra, as they are not needed
  for our analysis.}  The structure of quotients in this lattice must have
certain algebraic properties, which can be derived from \emph{tame
  congruence theory}~\cite{HobMcK88} and \emph{commutator
  theory}~\cite{FreMck87}. Bulatov constructs an algorithm for the
polynomial-time cases, based on decomposing this
congruence lattice and using the structure of its quotients. However,
he is only able to do this, in general, by transforming the relation
corresponding to the input instance to one which is a \emph{subdirect
  power}. It is even nontrivial to prove that this transformation
inherits the required property of the original. His paper runs to some
43 pages and is very difficult to follow for anyone who is not expert
in these areas.  The criterion of Bulatov's dichotomy is based
on infinite algebras constructed from $\Gamma\!$ and was not shown to be
decidable. It also seems difficult to apply it to recover
the special cases mentioned above.

Our main results are a new and elementary proof of Bulatov's theorem and a proof that the dichotomy is effective. Thus, we answer, in the affirmative, the major open question in~\cite{Bulato08}. We follow Bulatov's approach by working with the relation  over $\Gamma$ determined by the input, but we require almost no machinery from universal algebra. The little that is used is defined and explained below. We develop a different criterion for the \numcsp dichotomy, \emph{strong balance}, which is based on properties of ternary relations definable in the constraint language. We show that it is equivalent to Bulatov's \emph{congruence singularity} criterion.

Using strong balance, we construct a relatively simple iterative
algorithm for the polynomial-time cases, which requires no algebraic properties.
In fact, the bound on the time complexity of our counting algorithm is
no worse than that for deciding if the input has satisfying
assignments.

We then use our criterion to prove decidability of the \numcsp dichotomy. We show that deciding strong balance is in \np, where the input size is that of $\Gamma\!$. Of course, complexity is not a central issue in the nonuniform model of \numcsp, since $\Gamma$ is considered to be a constant. It is only decidability that is important. However, the complexity of deciding the dichotomy seems an interesting computational problem in its own right.



\subsection{Our proofs}

Our proofs are almost entirely self-contained and should be accessible
to readers with no knowledge of universal algebra and very little
background in \csp.  We use reductions from two previous papers on
counting complexity, by Dyer and Greenhill~\cite{DyeGre00} and by
Bulatov and Grohe~\cite{BulGro05}. We also use results from Bulatov
and Dalmau~\cite{BulDal07}, but we include short proofs of these.  The
papers~\cite{BulGro05,BulDal07} deal partly with ideas from universal
algebra, but we make no use of those ideas.  We use only one idea from
universal algebra, that of a \emph{\maltsev polymorphism}.  This will
be defined and explained in Section~\ref{sec:Defs} below.

The proof is based around a succinct representation for relations preserved by a \maltsev polymorphism. We call such relations \emph{strongly rectangular} for reasons which will become clear. Our representation is called a \emph{frame}, and is similar to the \emph{compact representation} of Bulatov and Dalmau~\cite{BulDal06}. Frames are smaller than compact representations, since they avoid some redundancy in the representation.

We define a \emph{frame} for a relation $R\subseteq D^n$ to be a
relation $F\subseteq R$ with the following two properties.  First,
whenever $R$ contains a tuple with $i$\/th component $a$, $F$ also
contains such a tuple.  Second, for $1< i\leq n$ say that a set
$S\subseteq D$ is $i$-equivalent in $R$ if $R$ contains tuples which
agree on their first $i-1$ elements and whose $i$th elements are
exactly the members of $S$.  Any set that is $i$-equivalent in $R$
must also be $i$-equivalent in $F$, but note that there may be several
common prefixes for $S$ in $R$ when only one is required in $F$.  We
show that every $n$-ary strongly rectangular relation over $D$ has a
\emph{small} frame of cardinality at most $|D|n$, whereas $R$ may have
cardinality up to $|D|^n\!$.  Further, we show how to construct such a
frame efficiently and how to recover a strongly rectangular relation
$R$ from any of its frames.

Now, suppose we have an instance $\Phi$ of $\ncsp$ for some strongly
rectangular constraint language $\Gamma\!$, with $m$ constraints in
$n$ variables.  Using methods similar to those of Bulatov and
Dalmau~\cite{BulDal06}, we construct a frame for the solution set of
$\Phi$ in polynomial time, by starting with a frame for $D^n$ and
introducing the constraints one at a time.  A frame is empty if, and
only if, it represents the empty relation so, at this point, we have
re-proven Bulatov and Dalmau's result that there is a polynomial-time
algorithm for the decision problem $\csp(\Gamma)$ for any strongly
rectangular constraint language $\Gamma\!$.   We give an explicit time
complexity for this algorithm, which is $\bigO(mn^4)$ for fixed
$\Gamma\!$. Bulatov and Dalmau~\cite{BulDal06} gave no time estimate,
showing only that their procedure runs in polynomial time.

Any ternary relation $R\subseteq A_1\times A_2\times A_3$ (where the
$A_i$ need not be disjoint) induces a matrix $M=(m_{xy})$ with rows
and columns indexed by $A_1$ and $A_2$ and with
\begin{equation*}
   m_{xy} = |\{z: (x,y,z)\in R\}|\,.
\end{equation*}
We say that $R$ is \emph{balanced} if $M$'s rows and columns can be
permuted to give a block-diagonal matrix in which every block has rank one,
and that a relation $R\subseteq D^n$ for any $n>3$ is balanced if
every expression of it as a ternary relation in $D^k\times
D^\ell\times D^m$ ($k, \ell, m\geq 1$, $k + \ell + m = n$) is
balanced.  A constraint language $\Gamma$ is \emph{strongly balanced}
if every relation of arity three or more that is pp-definable relation
over $\Gamma$ is balanced.  Via a brief detour through weighted
\numcsp{}, we show that \ncsp{} is \numpc{} if $\Gamma$ is not
strongly balanced.

If $\Gamma$ is strongly balanced, we compute the number of satisfying
assignments to a $\csp(\Gamma)$ instance as follows.  Let $R\subseteq
D^n$ be the set of satisfying assignments.  First, we construct a
small frame $F$ for $R$, as above.  If $R$ is unary, we have $F=R$ so
we can trivially compute $|R|$.

Otherwise, for $1\leq i< j\leq n$, let $N_{i,j}(a)$ be the number of
prefixes $u_1\dots u_i$ such that there is a tuple $u_1\dots
u_n\in R$ with $u_j = a$.  In particular, then, summing the values of
$N_{n-1,n}(-)$ gives $|R|$.  Since the functions $N_{1,j}$ can be
calculated easily from the frame, we
just need to show how to compute $N_{i, j}$ for each $j>i$, given
$N_{i-1,j}$ for each $j\geq i$.  Writing $[k]$ for the set $\{1, \dots,
k\}$, we can consider the set $\proj[{[i]} \cup \{j\}] R$ to be a
ternary relation on $\proj[{[i-1]}] R \times \proj[i] R \times
\proj[j] R$.  $R$ is strongly balanced so the matrix given by $M_{xy}
= |\{\bu: (\bu,x,y)\in \proj[{[i]} \cup \{j\}] R\}|$ is a rank-one block
matrix and the sum of the $a$-indexed column of the matrix is
$N_{i,j}(a)$.

By taking quotients with respect to certain
congruences, we obtain another rank-one block matrix $\widehat{M}$,
whose block structure and row and column sums we can determine.  A key
fact about rank-one block matrices is that this information is
sufficient to recover the entries of the matrix.  This allows us to
recover $M$ and, hence, compute the values $N_{i,j}(a)$ for each $j$
and $a$.  Iterating, we can determine the function $N_{n-1,n}$ and,
hence, compute $|R|$.

Finally, we show that the strong balance property is decidable. Our proof of decidability rests on showing that, if $\Gamma$ is not strongly balanced, then there is a counterexample with a number of variables that is only polynomial in the size of $\Gamma\!$. We do this by reformulating the strong balance criterion for a given formula $\Psi$ as a question concerning counting assignments in a formula derived from $\Psi$. This reformulation enables us to apply a technique of Lov\'asz~\cite{Lovasz67}. The technique further allows us to recast strong balance in terms of the symmetries of a fixed structure, that is easily computable from $\Gamma\!$. We are thus able to show that deciding strong balance is in \np, where the input size is that of $\Gamma\!$.


\subsection{Organisation of the paper}

The remainder of the paper is organised as follows.  Preliminary
definitions and notation are given in Section~\ref{sec:Defs}.  In
Section~\ref{sec:Rect}, we define the notion of strong rectangularity
that we use throughout the paper and, in
Section~\ref{sec:Strong-rect}, we further study the properties of
strongly rectangular relations and introduce frames, our succinct
representations of such relations.  We give an efficient procedure for
constructing frames in Section~\ref{sec:Frame}.  In
Section~\ref{sec:Counting}, we introduce counting problems and, in
Section~\ref{sec:Dichotomy}, we define the key notion of a strongly
balanced constraint language and prove that \ncsp is solvable in
polynomial time if $\Gamma$ is strongly balanced and is \numpc
otherwise. In Section~\ref{sec:Decide}, we show that our dichotomy is decidable, in fact in the complexity class \np.
Some concluding remarks appear in Section~\ref{sec:Conclude}.


\section{Definitions and notation}
\label{sec:Defs}

In this section, we present the definitions and notation used
throughout the paper.  We defer to Section~\ref{sec:Decide} material
relating to certain classes of functions that are used only in that
section.

For any natural number $n$, we write $[n]$ for the set $\set{1, \dots,
n}$.


\subsection{Relations and constraints}

Let $D=\set{d_1, d_2, \dots, d_q}$ be a finite \emph{domain} with
$q=|D|$.  We will always consider $q$ to be a constant and we assume
that $q\geq 2$ to avoid trivialities.  A
\emph{constraint language} $\Gamma$ is a finite set of finitary
relations on $D$, including the binary equality relation
$\set{(d_i,d_i): i\in[q]}$, which we denote by $\equ$.  We will
call $\frS=(D, \Gamma)$ a \emph{relational structure}.  We may view an
$r$-ary relation $H$ on $D$ with $\ell = |H|$ as an $\ell\times r$
matrix with elements in $D$.  Then a tuple $\bt \in H$ is any row of
this matrix.  We will usually write tuples in the standard notation,
for example $(t_1,t_2,\ldots,t_r)$. For brevity, however, we also
write tuples in string notation, for example, $t_1t_2\ldots t_r$, where
this can cause no confusion.

If $R$ is an $n$-ary relation and $\bi=(i_1, \dots, i_k)$ are distinct
elements of $[n]$, we write \proj[\bi]{R} for the \emph{projection} of
$R$ on $\bi$, the relation containing all tuples $(a_{i_1}, \dots,
a_{i_k})$ such that $(a_1, \dots, a_n)\in R$ for some values of the
$a_j$ where $j\notin \bi$. For $I\subseteq [n]$, we write
$\proj[I]{R}$ as shorthand for $\proj[\bi]{R}$, where $\bi$ is the
enumeration of $I$'s elements in increasing order.  For the relation
$\set{\bt}$, where $\bt$ is a single $n$-tuple, we write
\proj[\bi]{\bt} rather than \proj[\bi]{\set{\bt}}.

We define the \emph{size} of a relation $H$ as $\size{H} = \ell r$, the
number of elements in its matrix, and the size of $\Gamma$ as
$\size{\Gamma} = \sum_{H\in\Gamma}\size{H}$. To avoid trivialities,
we will assume that every relation $H\in \Gamma$ is nonempty, i.e.\@ that $\size{H}>0$. We will also assume that every $d\in D$ appears in a tuple of some relation $H\in\Gamma\!$. If this is not so for some $d$, we can remove it from $D$. It then follows that $\size{\Gamma}\geq q$.

Let $V = \set{\nu_1, \nu_2, \dots, \nu_n}$ be a finite
set of \emph{variables}.  An \emph{assignment} is a function $\bx\colon V\to
D$.  We will abbreviate $\bx(\nu_i)$ to $x_i$.  If $\set{i_1, i_2,
\dots, i_r}\subseteq[n]$, we write $H(x_{i_1}, x_{i_2}, \dots,
x_{i_r})$ for the relation $\con=\set{\bx:(x_{i_1}, x_{i_2}, \dots,
x_{i_r})\in H}$ and we refer to this as a \emph{constraint}.  Then
$(\nu_{i_1}, \nu_{i_2}, \dots, \nu_{i_r})$ is the \emph{scope} of the
constraint and we say that $\bx$ is a \emph{satisfying} assignment
for the constraint if $\bx\in \con$.

A \emph{$\Gamma$-formula} $\Phi$ in a set of variables $\set{x_1, x_2,
\dots, x_n}$ is a conjunction of constraints $\con_1 \wedge \cdots
\wedge \con_m$.  We will identify the variables with the $x_i$ above,
although strictly they are only a \emph{model} of the formula.  Note that the precise labelling of the variables in $\Phi$ has no real significance. A formula remains the same if its variables are relabelled under a bijection to any other set of variable names.

A $\Gamma$-formula $\Phi$ describes an instance of the \emph{constraint
satisfaction problem} (\csp) with \emph{constraint language} $\Gamma\!$.
A satisfying assignment for $\Phi$ is an assignment that satisfies all
$\con_i$ ($i\in[m]$).  The set of all satisfying assignments for $\Phi$
is the \emph{$\Gamma$-definable} relation $R_\Phi$ over $D$.  We
make no distinction between $\Phi$ and $R_\Phi$, unless this could
cause confusion.



\subsection{Definability}

A \emph{primitive positive} (pp) formula $\Psi$ is a $\Gamma$-formula
$\Phi$ with existential quantification over some subset of the
variables.  A satisfying assignment for $\Psi$ is any satisfying
assignment for $\Phi$.  The unquantified (free) variables then
determine the \emph{pp-definable} relation $R_\Psi$, a projection of
$R_\Phi$.  Note that any permutation of the columns of a pp-definable
relation is, itself, pp-definable.  Again, we make no distinction
between $\Psi$ and $R_\Psi$.

The set of all $\Gamma$-definable relations is denoted by
$\csp(\Gamma)$ and the set of all pp-definable relations is the
\emph{relational clone} \cocl.  If $\Gamma = \set{H, \equ}$, we just
write \cocl[H].  An \emph{equivalence relation} in \cocl is
called a \emph{congruence}.


\subsection{Polymorphisms}

A $k$-ary \emph{polymorphism} of $\Gamma$ is any function $\psi\colon
D^k\to D$, for some $k$, that preserves all the relations in
$\Gamma\!$.  By this we mean that, for every $r$-ary relation
$H\in\Gamma$ and every sequence $\bu_1, \dots, \bu_k$ of $r$-tuples in
$H$,
\begin{equation*}
   \psi(\bu_1,\bu_2,\dots,\bu_k)
        = \big(\psi(u_{1,1}, \dots, u_{k,1}),
               \psi(u_{1,2}, \dots, u_{k,2}),
               \, \dots,\,
               \psi(u_{1,r}, \dots, u_{k,r})
          \big)\in H\,.
\end{equation*}
It is well known that any polymorphism of $\Gamma$ preserves all
relations in $\cocl$ (see Lemma~\ref{lem15}).

A \emph{\maltsev polymorphism} of $\Gamma$ is a polymorphism
$\ph\colon D^3\to D$ such that, for all $a,b\in D$, $\ph(a,b,b) = \ph(b,b,a) = a$.  (So, in particular, $\ph(a,a,a) = a$.)  We will usually
present calculations using \ph in a four-row table.  The first three
rows give the triple of ``input'' tuples $\bt_1,\bt_2,\bt_3$ and the
fourth gives the ``output'' $\ph(\bt_1,\bt_2,\bt_3)$.  For example,
the  table below indicates that $\ph(a\bu, a\bv, b\bw)=(b, \ph(\bu, \bv, \bw))$.
\begin{equation*}
   \begin{array}{c@{\hspace{5mm}}c}
        a & \bu \\
        a & \bv \\
        b & \bw \\ \hline
        b & \ph(\bu,\bv,\bw)\,.
    \end{array}
\end{equation*}


\subsection{Complexity}

For any alphabet $\Sigma$, we denote by \fp{} the class of
functions $f\colon \Sigma^*\to \nats$ for which there is a
deterministic, polynomial-time Turing machine that, given input
$x\in\Sigma^*\!$, writes $f(x)$ (in binary) to its output tape.
\nump{} is the class of functions $f\colon \Sigma^*\to \nats$ for which some  nondeterministic, polynomial-time Turing machine has exactly $f(x)$ accepting computations for every input $x\in\Sigma^*\!$.

Completeness for \nump is defined with respect to polynomial-time Turing
reductions~\cite{Valian79b}, also known as \emph{Cook reductions}.
For functions $f,g\colon \Sigma^*\to \nats$, a \emph{polynomial-time
  Turing reduction} from $f$ to $g$ is a polynomial-time oracle Turing
machine that can compute $f$ using an oracle for $g$.  A function
$f\in\nump$ is \emph{\nump{}-complete} if there is a Cook reduction to
$f$ from every problem in \nump.

The class \nump{} plays a role in the complexity of counting problems analogous to that played by \np{} in decision problems.  Note, however, that,
subject to standard complexity-theoretic assumptions, \nump{}-complete
problems are much harder than \np{}-complete problems. Toda has
shown that $\ptime^{\nump}$ includes the whole of the polynomial-time
hierarchy~\cite{Toda89}, whereas $\ptime^{\np}$ is just the
hierarchy's second level.


\section{Rectangular relations}
\label{sec:Rect}

A binary relation $B\subseteq A_1\times A_2$ is called
\emph{rectangular} if $(a,c), (a,d), (b,c)\in B$ implies $(b,d)\in B$
for all $a,b\in A_1$, $c,d\in A_2$.  We may view $B$ as an undirected
bipartite graph $\cG_B$, with vertex bipartition $A_1$, $A_2$ and edge
set $E_B=\set{\set{a_1,a_2}: (a_1,a_2)\in B}$.  Note that we do not
insist that $A_1\cap A_2=\emptyset$ but, if $a\in A_1\cap A_2$, $a$ is
regarded as labelling two distinct vertices, one in $A_1$ and one in
$A_2$. Formally, $A_1$ and $A_2$ should be replaced by the disjoint
vertex sets $\{1\}\times A_1$ and $\{2\}\times A_2$ but this would
unduly complicate the notation. We will assume that $\proj{B}=A_i$
($i=1,2$), so that $\cG_B$ has no isolated vertices. The connected components
of $\cG_B$ will be called the \emph{blocks} of $B$.

Rectangular relations have very simple structure.
\begin{lemma}
\label{lem10}
    If $B$ is rectangular, $\cG_B$ comprises $k$ bipartite cliques,
    for some $k\leq\min\set{|A_1|,|A_2|}$.
\end{lemma}
\begin{proof}
    Let $k$ be the number of connected components of $\cG_B$.  Every
    vertex is included in an edge so $k\leq \min\set{|A_1|, |A_2|}$.
    Consider any component $C$ and suppose it is not a bipartite
    clique.  Let $a\in A_1\cap C$, $z\in A_2\cap C$ be such that
    $\set{a,z}\notin
    E_B$.  Thus, a shortest path in $C$ from $a$ to $z$ has length at
    least~$3$.  If $a$, $b$, $c$, $d$ are the first four vertices on
    such a path, then $\set{a,b}, \set{b,c}, \set{c,d}\in E_B$, but
    $\set{a,d}\notin E_B$ as, otherwise, there would be a shorter
    path from $a$ to $z$.  But this is equivalent to $(a,b), (c,b),
    (c,d)\in B$ and $(a,d)\notin B$, contradicting rectangularity.\qquad
\end{proof}

Where appropriate, we do not distinguish between $B$ and $\cG_B$.  For example, we will refer to a connected component of $\cG_B$ as a block.

\begin{corollary}
\label{cor10}
    The relations
    \begin{equation*}
       \theta_1(x_1, x_2) \equiv \exists y\,\big(
                                         B(x_1, y)\wedge B(x_2, y)\big)
        \quad \text{and} \quad
        \theta_2(y_1, y_2) \equiv \exists x\,\big(
                                         B(x, y_1)\wedge B(x, y_2)\big)
    \end{equation*}
    are equivalence relations on $\proj[1]{B}$, $\proj[2]{B}$
    respectively.  The equivalence classes of $\theta_1$ and $\theta_2$
    are in one-to-one correspondence.
\end{corollary}
\begin{proof}
    The blocks of $B$ induce partitions of $A_1$ and $A_2$ which are
    in one-to-one correspondence.  These clearly define the equivalence
    classes of $\theta_1$ and $\theta_2$.\qquad
\end{proof}

\begin{corollary}
\label{cor15}
    If $\Gamma$ is a constraint language and $B\in\cocl$ is
    rectangular, then the relations $\theta_1$ and $\theta_2$ of
    Corollary~\ref{cor10} are congruences in \cocl.
\end{corollary}
\begin{proof}
    Since $B$ has a pp-definition, so too do $\theta_1$ and $\theta_2$.
\qquad\end{proof}

We say that a relation $R\subseteq D^n$ for $n\geq 2$ is rectangular
if every expression of $R$ as a binary relation in $D^k\times D^{n-k}$
($1\leq k < n$) is rectangular.  We call a constraint
language $\Gamma$ \emph{strongly rectangular} if every relation
$B\in\cocl$ of arity at least~2 is rectangular.  If $R\subseteq D^n$
is a relation, we say that it is strongly rectangular if \cocl[R] is
strongly rectangular.  If $R\in\cocl$ for a strongly rectangular
$\Gamma\!$, then $R$ is strongly rectangular, since
$\cocl[R]\subseteq\cocl$.

From the definition, it is not clear whether the strong
rectangularity of $\Gamma$ is even decidable, since \cocl is an
infinite set.  However, it is decidable, as we will now show. The following result is usually proven in an algebraic setting. That proof is not difficult, but requires an understanding of concepts from universal algebra, such as \emph{free algebras} and \emph{varieties}~\cite{DenWis02}. Therefore, we will give a proof in the relational setting. Moreover, we believe that this proof will provide rather more insight for the reader whose primary interest is in relations.

First, we require the following lemma, which is well-known from the
folklore; we provide a proof for completeness.

\begin{lemma}
\label{lem15}
    $\ph$ is a polymorphism of $\Gamma$ if, and only if, it is a
    polymorphism of $\cocl$.
\end{lemma}
\begin{proof}
    Let $\ph$ be a polymorphism of $\Gamma$ and let $R\in\cocl$.  We
    prove that $\ph$ is a polymorphism of $R$ by induction on the
    structure of the defining formula of $R$.  The base case, atomic
    formulae ($H(\bx)$ for relations $H\in\Gamma$) is trivial.

    Suppose $R$ is defined by $\exists y\,\psi(\bx,y)$.  If $\ba_1,
    \ba_2, \ba_3\in R$, then there are $b_1, b_2, b_3$ such that
    $\ba_i b_i\in \psi$ ($i\in \set{1,2,3}$).  If $\ph$ is a
    polymorphism of $\psi$, then it follows that $\bc d = \ph(\ba_1
    b_1, \ba_2 b_2, 
    \ba_3 b_3)\in \psi$, which means that $\bc\in R$, as required.

    Finally, suppose $R$ is defined by $\psi(\bx) \wedge \chi(\bx)$.
    If $\ba_1, \ba_2, \ba_3\in R$, then $\ba_i\in \psi \cap \chi$ for
    each $i$.  If $\ph$ is a polymorphism of $\psi$ and of $\chi$ then
    $\bc = \ph(\ba_1, \ba_2, \ba_3)\in \psi\cap \chi$ and, therefore,
    $\bc\in R$.

    Conversely, $\Gamma\subseteq \cocl$ so every polymorphism of
    $\cocl$ is a polymorphism of $\Gamma$.
\end{proof}

\begin{lemma}
\label{lem20}
    A constraint langauge $\Gamma$ is strongly rectangular if, and
    only if, it has a \maltsev polymorphism.
\end{lemma}
\begin{proof}
    Suppose $\Gamma$ has a \maltsev polymorphism \ph. Consider any
    pp-definable binary relation $B\subseteq D^r\times D^s\!$.  By
    Lemma~\ref{lem15}, $\ph$ is also a polymorphism of $B$.  If
    $(\ba,\bc), (\ba,\bd), (\bb,\bd)\in B$ then we have
    $(\ph(\ba,\ba,\bb), \ph(\bc,\bd,\bd)) = (\bb,\bc)\in B$, from the
    definition of a \maltsev polymorphism.  Thus, $B$ is rectangular
    and hence $\Gamma$ is strongly rectangular.

Conversely, suppose $\Gamma$ is strongly rectangular. Denote the relation $H\in\Gamma$ by $H=\set{\bu^H_i:i\in[\ell_H]}$, where $\bu^H_i\in D^{r_H}\!$. Consider the $\Gamma$-formula
\begin{equation*}
    \Phi(\bx)\ =\ %
        \bigwedge_{H\in\Gamma}
            \bigwedge_{i_1\in[\ell_H]}
                \bigwedge_{i_2\in[\ell_H]}
                    \bigwedge_{i_3\in[\ell_H]}H
                        \big(\bx^H_{i_1,i_2,i_3}\big)\,,
\end{equation*}
where $\bx^H_{i_1,i_2,i_3}$ is an $r_H$-tuple of variables, distinct for all $H\in\Gamma\!$, $i_1,i_2,i_3\in[\ell_H]$. Thus, the relation $R_\Phi$ has arity
$r_\Phi=\sum_{H\in\Gamma} r_H\/\ell_H^3$ and $|R_\Phi|=\prod_{H\in\Gamma} {\ell_H}^{\ell_H^3}\!$.

Clearly $R_\Phi$ has three tuples $\bu_1$, $\bu_2$, $\bu_3$ such that
the sub-tuple of $\bu_j$ corresponding to $\bx^H_{i_1,i_2,i_3}$ is
$\bu^H_{i_j}$ for each $j\in\set{1,2,3}$ and each $i_1, i_2, i_3\in
[\ell_H]$.  Then $U = \set{\bu_1, \bu_2,
\bu_3}$ has the following universality property for $\Gamma\!$. For
all $H\in\Gamma$ and every triple of (not necessarily distinct) tuples
$\bt_1$, $\bt_2$, $\bt_3\in H$, there is a set $I = I(\bt_1,\bt_2,\bt_3)$
with $I \subseteq [r_\Phi]$, $|I|=r_H$ such that $\proj[I]{R_\Phi}=H$
and $\proj[I]{\bu_i}=\bt_i$ ($i=1,2,3$).

Now, for each set of identical columns in $U\!$, we impose equality on the corresponding variables in $\Phi$, to give a $\Gamma$-formula $\Phi'\!$. Let $U'$ be the resulting submatrix of $U\!$, with rows $\bu'_1$, $\bu'_2$, $\bu'_3$. Observe that $U'$ is obtained by deleting copies of columns in $U\!$. Therefore $U'$ has no identical columns and has a column $(a,b,c)$ for all $a,b,c\in \proj[k]{H}$ with $H\in\Gamma$ and $k\in[r_H]$.

Next, for all columns $(a,b,c)$ of $U'$ such that $b\notin\set{a,c}$, we impose existential quantification on the corresponding variables in $\Phi'\!$, to give a pp-formula $\Phi''\!$. Let $U''$ be the submatrix of $U'$ with rows $\bu''_1$, $\bu''_2$, $\bu''_3$ corresponding to $\bu'_1$, $\bu'_2$, $\bu'_3$. Then $U''$ results from deleting columns in $U'$ and $U''$ has columns of the form $(a,a,b)$ or $(c,d,d)$. Thus, after rearranging columns (relabelling variables), we will have
\begin{equation*}
   U''\ =\ \begin{bmatrix} \, \bu''_1 \, \\[1pt] \, \bu''_2 \, \\[1pt] \, \bu''_3 \, \end{bmatrix}\ =\ \begin{bmatrix} \, \ba & \bc \, \\[1pt] \, \ba & \bd \, \\[1pt] \, \bb & \bd \, \end{bmatrix}\,,
\end{equation*}
for some nonempty tuples $\ba,\,\bb,\,\bc,\,\bd$ so, by strong rectangularity,
$\bu''\,=\,\begin{bmatrix} \bb &\!\! \bc\end{bmatrix}\in R_{\Phi''}$.

Removing the existential quantification in $\Phi''\!$, $\bu''$ can be extended to $\bu'\in R_{\Phi'}$. Now, if column $k$ of $U'$ is $(a,b,c)$ say, we define $\ph(a,b,c)=u'_k$. This is unambiguous, since $U'$ has no identical columns. Thus, $\bu'=\ph(\bu'_1,\bu'_2,\bu'_3)\in R_{\Phi'}$. If, for any $a,b,c\in D$, $\ph(a,b,c)$ remains undefined, we will set $\ph(a,b,c)=a$ unless $a=b$, in which case $\ph(a,b,c)=c$. Clearly \ph satisfies $\ph(a,b,b)=\ph(b,b,a)=a$, for all $a,b\in D$, and so has the \maltsev property.

Removing the equalities between variables in $\Phi'\!$, $\bu'$ can be
further extended to give the tuple $\bu=\ph(\bu_1,\bu_2,\bu_3)\in R_\Phi$. This is
consistent since $\bu$ satisfies the equalities imposed on $\Phi$ to
give $\Phi'\!$. Now, for any $\bt_1,\bt_2,\bt_3\in H$, the
universality property of $U$ implies that, for some $I$,
$\proj[I]{\bu}=\ph(\bt_1,\bt_2,\bt_3)\in H$. Thus, \ph preserves all
$H\in\Gamma\!$, so it is a polymorphism and hence a \maltsev
polymorphism.\qquad
\end{proof}

\begin{remark}\label{rem04}%
    Observe that the proof of Lemma~\ref{lem20} uses all the elements of pp-definability. Thus, if Lemma~\ref{lem20} is to hold true, the definition of strong rectangularity cannot be significantly weakened.
\end{remark}

\begin{remark}\label{rem06}%
    The proof of Lemma~\ref{lem20} is constructive and, hence, implies an algorithm for deciding whether $\Gamma$ is strongly rectangular and, if so, determining a \maltsev polymorphism \ph. However, we describe a more efficient method in Lemma~\ref{lem30} below.
\end{remark}

Note that strong rectangularity is invariant under permutations of the
columns of a relation, both by Lemma~\ref{lem20} (since permutations
of columns do not affect \maltsev polymorphisms) and by the fact that
permutations are pp-definable.  We will use this fact
repeatedly and consider a relation $R\subseteq D^n$ for some $n>2$ to
be a binary relation on $D^k\times D^{n-k}$ or a ternary relation on
$D^k\times D^\ell\times D^{n-k-\ell}\!$, for any appropriate values of
$k$ and $\ell$.

In the algebraic setting, the result corresponding to
Lemma~\ref{lem20} is that \cocl has a \maltsev polymorphism if, and
only if, $\Gamma$ is \emph{congruence permutable}.  See, for
example,~\cite{DenWis02}.  This has the following meaning.  If
$\rho_1$ and $\rho_2$ are congruences on a pp-definable set
$A\subseteq D^r\!$, define the \emph{relational product}
$\psi=\rho_1\circ\rho_2$ by $\psi(\bx,\by) = \exists \bz\,
\big(\chi(\bz)\wedge \rho_1(\bx,\bz)\wedge \rho_2(\bz,\by)\big)$,
where $\chi$ is the formula defining $A$.  Then $\rho_1,\,\rho_2$ are
\emph{permutable} if $\psi(\bu, \bv)$ implies $\psi(\bv,\bu)$ for all
$\bu, \bv\in A$ or, equivalently, $\rho_1\circ\rho_2 =
\rho_2\circ\rho_1$.  Now $\Gamma$ is congruence permutable if every
pair of congruences on the same set $A$ is permutable.  For
completeness, we will prove the following.

\begin{lemma}\label{lem25}
    $\Gamma$ is strongly rectangular if, and only if, it is congruence permutable.
\end{lemma}
\begin{proof}
    Suppose $\Gamma$ is strongly rectangular. If $\rho_1$, $\rho_2$
    are congruences on a pp-definable set $A\subseteq D^r\!$, let
    $\psi$ be the relational product, as defined above.  Clearly
    $\psi$ is a pp-definable binary relation on $D^r\!$. Then, if
    $(\bu,\bv)\in \psi$, we have $(\bu,\bu),\, (\bu,\bv),\,
    (\bv,\bv)\in\psi$, since $\rho_1$ and $\rho_2$ are congruences.
    But this implies $(\bv,\bu)\in\psi$ since $\psi$ is rectangular.
    Thus, $\Gamma$ is congruence permutable.

Conversely, if $\Gamma$ is congruence permutable, consider a
pp-definable relation $B\subseteq D^r\times D^s\!$.  Define a relation
$\tw[1]$ on $B$ by $(\bx_1,\by_1)\tw[1](\bx_2,\by_2)$ if, and only if,
$(\bx_1,\by_1)\in B$, $(\bx_2,\by_2)\in B$ and $\bx_1=\bx_2$. This is
pp-definable, by $B(\bx_1,\by_1)\wedge B(\bx_2,\by_2)\wedge
(\bx_1=\bx_2)$, and is clearly an equivalence relation. Hence it is a
congruence. Similarly, define a congruence $\tw[2]$ on $D^{r+s}\!$ by
$(\bx_1,\by_1)\tw[2](\bx_2,\by_2)$ if, and only if,
$(\bx_1,\by_1),\,(\bx_2,\by_2)\in B$ and $\by_1=\by_2$. Let
$\psi\,=\,\tw[1]\circ\tw[2]$.

Suppose $\big((\ba,\bc),(\bb,\bd)\big)\in\psi$. Then there exists $(\bu,\bv)\in B$ such that $(\ba,\bc)\tw[1](\bu,\bv)\tw[2](\bb,\bd)$. Thus, $(\bu,\bv)=(\ba,\bd)$ and, hence,
$(\ba,\bc),\,(\ba,\bd),\,(\bb,\bd)\in B$. Congruence permutability implies
$\big((\bb,\bd),(\ba,\bc)\big)\in\psi$. Hence there exists $(\bu'\!,\bv')\in B$ such that $(\bb,\bd)\tw[1](\bu'\!,\bv')\tw[2](\ba,\bc)$. Thus, $(\bu'\!,\bv')=(\bb,\bc)$. Therefore we have $(\bb,\bc)\in B$ and $\Gamma$ is strongly rectangular.\qquad
\end{proof}
\begin{corollary}\label{cor20}
  $\Gamma$ is congruence permutable if, and only if, it has a
    \maltsev polymorphism.
\end{corollary}
\begin{proof}
   This follows directly from Lemmas~\ref{lem20} and~\ref{lem25}.\qquad
\end{proof}

We will now consider the complexity of deciding whether $\Gamma$ is strongly rectangular.
\begin{lemma}
\label{lem30}
    We can decide whether $\Gamma$ is strongly rectangular in
    $\bigO(\size{\Gamma}^4)$ time and, if so, determine a \maltsev
    polymorphism \ph.
\end{lemma}
\begin{proof}
    Observe that there are at most $q^{q(q-1)^2}$ possible \maltsev
    operations $D^3\to D$.  This follows since there are $q(q-1)^2$
    triples $a,b,c\in D$ which have $b \notin\set{a,c}$. For all other
    triples, the value of $\ph(a,b,c)$ is determined by the condition
    that \ph is \maltsev{}.  Thus, there are $\bigO(1)$ possibilities
    for \ph.  For an $r$-ary relation $H\in\Gamma$ with $\ell$
    tuples, we can check in $\bigO(\ell^4 r) = \bigO(\size{H}^4)$ time whether
    $H$ is preserved by any of them.  If so, we have \ph\!; if not,
    $\Gamma$ is not strongly rectangular.\qquad
\end{proof}
\begin{remark}\label{rem80}%
    We have assumed that $q$ is a constant in Lemma~\ref{lem30}. We revisit this question in Section~\ref{sec:Decide}, where we make no such assumption.
\end{remark}

In view of Lemma~\ref{lem30}, we may assume that we have determined a
\maltsev polymorphism \ph for any given strongly rectangular
$\Gamma\!$.

Strongly rectangular constraint languages have another useful
property.  For each $a\in D$, define the \emph{constant relation}
$\chi_a=\set{(a)}$.  Then the constraint $\chi_a(x_i)$ fixes the value
of $x_i$ to be $a$.

\begin{lemma}
\label{lem35}
    If $\Gamma$ is strongly rectangular, then so is $\Gamma' = \Gamma \cup
    \set{\chi_a}$.
\end{lemma}
\begin{proof}
    By Lemma~\ref{lem20}, $\Gamma$ is preserved by a \maltsev
    polymorphism \ph.  Since $\ph(a,a,a)=a$ for any $a\in D$, \ph also
    preserves $\chi_a$.  Thus \ph preserves $\Gamma'\!$, so $\Gamma'$
    is strongly rectangular, by Lemma~\ref{lem20}.\qquad
\end{proof}

In the light of Lemma~\ref{lem35}, we may assume that $\set{\chi_a:
a\in D}\subseteq \Gamma$ whenever $\Gamma$ is strongly rectangular.

\begin{remark}\label{rem90}%
    More generally, the property of a polymorphism $\psi$ that we have used in Lemma~\ref{lem35}, that $\psi(x,x,\ldots,x)=x$ for any $x\in D$, is called \emph{idempotence} in the algebraic literature on \csp.
\end{remark}


\section{The structure of strongly rectangular relations}
\label{sec:Strong-rect}

Let $R\subseteq D^n$ be a strongly rectangular relation.  For any $i\in[n]$, we say that an $n$-tuple $\bt\in R$ is a \emph{witness} for $a\in\proj{R}$ if
$t_i=a$.  We will abbreviate this by saying that $\bt$ witnesses
$(a,i)$.  If $\bt = (\bu,a,\bv)\in R$, we call $\bu$ a \emph{prefix}
for $a$.  Now define a relation \tw on $\proj{R}$ by $a\tw b$ if, and
only if, there exists $\bu\in D^{i-1}$ which is a common prefix for
$a$ and $b$.  That is, there exist $\bv_a,\bv_b\in D^{n-i}$ such that
$(\bu,a,\bv_a),\, (\bu,b,\bv_b)\in R$.

\begin{lemma}\label{lem40}
    \tw is an equivalence relation on \proj{R} and a congruence in
    \cocl[R].
\end{lemma}
\begin{proof}
    Consider the binary relation $B$ on $\proj[{[i-1]}]{R} \times
    \proj{R}$ defined by $B(\bu,a) = \exists\by\,R(\bu,a,\by)$.  Then
    \tw is the equivalence relation $\theta_2$ of
    Corollary~\ref{cor10}, which is a congruence by
    Corollary~\ref{cor15}.\qquad
\end{proof}

Let $\cE_{i,k}$ ($k\in[\kappa_i]$) be the equivalence classes of \tw
for $\kappa_i\in[q]$, $i\in[n]$.  Observe that $\kappa_1=1$, since all
$a\in\proj[1]{R}$ have witnesses with the common empty prefix.  More
generally, we make the following observation, which follows directly
from the block structure of the relation $B$ in the proof of Lemma~\ref{lem40}.

\begin{corollary}
\label{cor30}
    There is a common prefix $\bu_{i,k}\in D^{i-1}$ for all $a\in
    \cE_{i,k}$ ($k\in[\kappa_i], i\in[n]$) and we can choose
    $\bu_{i,k}$ to be any prefix of any $a\in \cE_{i,k}$.
\end{corollary}

Following Bulatov and Dalmau~\cite{BulDal06}, if $H$ is any relation
and \ph a \maltsev operation (i.e., a ternary function that is not
necessarily a polymorphism but has the property that $\ph(a,b,b) =
\ph(b,b,a) = a$ for all $a,b\in D$), then \cl{H} is the smallest
relation that contains $H$ and is closed under \ph.  Clearly
\cl{H} is a strongly rectangular relation with polymorphism \ph and we
say that the $H$ \emph{generates} \cl{H}.  The following observation,
from~\cite{BulDal06}, gives a simple but important fact.

\begin{lemma}
\label{lem50}
    Let $H$ be an $n$-ary relation.  If $I \subseteq [n]$, then
    $\cl{\proj[I]{H}} = \proj[I]{\cl{H}}$.
\end{lemma}
\begin{proof}
    Consider generating $\cl{\proj[I]{H}}$ while retaining all $n$
    columns of $H$.  Each row of the resulting $n$-ary relation will
    be in \cl{H}, so we have $\cl{\proj[I]{H}} \subseteq
    \proj[I]{\cl{H}}$.  But further operations to generate \cl{H}
    cannot add new rows to \cl{\proj[I]{H}}.  So, in fact, we have
    $\cl{\proj[I]{H}} = \proj[I]{\cl{H}}$.\qquad
\end{proof}

Let $S=\set{\bt_1, \bt_2, \dots, \bt_s}$ be a set of $n$-tuples,
presented as an $s\times n$ matrix.  If $I\subseteq[n]$, we will need
to compute a relation $T\subseteq \cl{S}$ such that $\proj[I]{T} =
\cl{\proj[I]{S}} = \proj[I]{\cl{S}}$.

\begin{lemma}
\label{lem60}
    If $\ell=|\proj[I]{\cl{S}}|$ and $s=|S|$, then a relation
    $T\subseteq \cl{S}$ such that $\proj[I]{T} = \proj[I]{\cl{S}}$ can
    be computed in time $\bigO(n\ell^3+s\ell^4)$.
\end{lemma}
\begin{proof}
    Consider the algorithm \closure, on the following page.

    \begin{algorithm}
    \floatname{}
    \caption{\textbf{procedure} $\closure(I)$}
    \begin{algorithmic}[1]
    \STATE $\ell\gets s$, $j_1\gets 2$
    \WHILE{$j_1\leq \ell$}
        \FOR{$j_2\in [j_1]$}
            \FOR{$j_3\in [j_2]$}
                \FORALL{permutations $(k_1,k_2,k_3)$ of $\set{j_1,j_2,j_3}$
                        such that $k_2\notin \{k_1,k_3\}$}
                    \STATE $\bu\gets \ph(\bt_{k_1},\bt_{k_2},\bt_{k_3})$
                                                          \label{line10}
                    \IF{there is no $j\in[\ell]$ such that $\proj[I]\bt_j =
                        \proj[I]\bu$}                     \label{line20}
                    \STATE $\ell\gets \ell+1$, $\bt_\ell\gets \bu$
                    \ENDIF
                \ENDFOR
            \ENDFOR
        \ENDFOR
        \STATE $j_1\gets j_1+1$
    \ENDWHILE
    \end{algorithmic}
    \end{algorithm}

    The correctness of \closure is trivial.  At termination, all
    $\ell^3$ triples $(k_1,k_2,k_3)\in[\ell]^3$ have
    been considered for generating new $n$-tuples (in
    line~\ref{line10}), so we have computed \cl{\proj[I]{S}}.  The
    analysis is equally easy.  There are $\ell^3$ triples
    $(k_1,k_2,k_3)$.  For each triple, the generation in
    line~\ref{line10} takes $\bigO(n)$ time and the search in
    line~\ref{line20} requires $\bigO(s\ell)$ time, with the obvious
    implementations.  Thus, the total is $\bigO(n\ell^3+s\ell^4)$.\qquad
\end{proof}

The procedure outlined in~\cite{BulDal06} has complexity $\bigO(n\ell^4 +
s\ell^5)$, since the same triple $(k_1,k_2,k_3)$ can appear
$\Omega(\ell)$ times.  The procedure \closure simply avoids this.

The time complexity of \closure could be improved, for example, by using a more sophisticated data structure to implement the searches in
line~\ref{line20}.  However we do not pursue such issues here, or
elsewhere in the paper.

Now we define a \emph{frame} for an $n$-ary relation $R$ to be a set
$F\subseteq R$ such that
\setlist{nolistsep}
\begin{enumerate}[label=(\alph*)]
\item \label{item10}
    $\proj[i] F = \proj[i] R$ for each $i\in [n]$; and
\item \label{item20}
    there is a $\bv_{i,k}\in D^{i-1}$ for each equivalence class $\cE_{i,k}$ of \tw ($k\in[\kappa_i], i\in[n]$) such that, for each $a\in\cE_{i,k}$, there exists a $\bw_a\in F$ with $\proj[{[i]}]\bw_a=\bv_{i,k}a$.
\end{enumerate}

Clearly, $R$ itself satisfies the definition of a frame, so every
relation has at least one frame.  However, we will show that strongly
rectangular relations have frames that can be much smaller than $R$
and we call a frame for a strongly rectangular relation $R\subseteq
D^n$ \emph{small} if $|F|\leq n(q-1)+1$.

A \emph{witness function} for a frame $F$ of the relation $R$ is a
function $\wit\colon D\times [n]\to F$ such that $\wit(a,i)$ witnesses
$(a,i)$ for all $a\in \proj R$ and $i\in[n]$ and $\proj[{[i-1]}]
\wit(a,i) = \proj[{[i-1]}] \wit(b,i)$ when $a\tw b$.  That is,
$\wit(a,i)$ returns a witness for $(a,i)$ and, if $(a_1,i)$, $\dots$, $(a_k,i)$
have witnesses with a common prefix, then $\wit$ returns such witnesses.

\begin{lemma}
\label{lem70}
    Let $F$ be a frame for a strongly rectangular relation $R\subseteq
    D^n\!$.  We can determine a small frame $F'$ for $R$ and a
    surjective witness function $\wit'\colon D\times[n]\to F'$ in time
    $\bigO(\size{F}^2)$.
\end{lemma}
\begin{proof}
    In time $\bigO(\size{F})^2\!$, we can compute the relations \tw ($i\in
    [n]$) and common prefixes for each \tw{}-equivalence class.
    Hence, we can compute a witness function $\wit$ for $F$.  Further,
    we may delete from $F$ any tuple $\bt$ for which $\wit^{-1}(\bt) =
    \emptyset$.  Because $\wit$ is a witness function, the resulting
    set is still a frame for $R$ and has size at most $\sum_{i\in [n]}
    |\proj[i] R| \leq nq$.

    Now we construct $F'$ and $\wit'$ as follows.  Choose any $\Bf\in
    F$ and set $F'=\set{\Bf}$.  Then, for each $i\in[n]$, do the
    following.  Let $g=\wit(f_i, i)$ and set $\wit'(f_i,i)\gets \Bf$.
    Now, consider in turn each $a\neq f_i$ such that $a\tw f_i$ and let
    $\bh = \wit(a,i)$.
    Note that $\bg$ and $\bh$ have the same prefix $\bu'\in
    D^{i-1}\!$, since $F$ is a frame, and suppose $\Bf$ has prefix
    $\bu\in D^{i-1}\!$. Then set $\bh'\gets \ph(\Bf, \bg, \bh)$,
    $F'\gets F' \cup \set{\bh'}$ and $\wit'(a,i)\gets \bh'\!$.  Since
    \begin{equation*}
        \begin{array}{c@{\ \ :\quad}c@{\hspace{1cm}}c@{\hspace{0.5cm}}c}
            \Bf\phdash & \bu\phdash & f_i & \bv_{\phantom{a}} \\
            \bg\phdash & \bu'       & f_i & \bv'_{\phantom{a}} \\
            \bh\phdash & \bu'       & a   & \bv_a \\\hline
            \bh'       & \bu\phdash & a   & \ph(\bv,\bv'\!,\bv_a)\,,
        \end{array}
    \end{equation*}
    this ensures that $F'$ retains property~\ref{item20} of a frame.
    Having performed these steps for each $i\in[n]$, we deal with
    those $a\in\proj{F}$ with $a\not\tw f_i$ by setting $F'\gets F'
    \cup \set{\wit(a,i)}$ and $\wit'(a,i)\gets\wit(a,i)$.

    The final size of $F'$ can be bounded as follows.  The tuple $\Bf$
    witnesses $(f_i,i)$ for all $i\in[n]$.  Then, for each $i\in[n]$,
    there is at most one tuple in $F'$ witnessing $(a,i)$ for each
    $a\in\proj{R} \setminus \set{f_i}$.  Since there are, in total,
    $\sum_{i=1}^n \big(|\proj{R}|-1\big)\leq n(q-1)$ such pairs
    $(a,i)$, it follows that $F'$ is a small frame.

    The time bound is easy.  Given the function $\wit$, we can
    determine the $\bh'$ in $\bigO(n)$ for each $i\in[n]$.  All other
    operations require $\bigO(1)$ time for each $i\in[n]$.  Thus, we can
    need only $\bigO(n^2) = \bigO(\size{F}^2)$ time once we have determined
    $\wit$, which can also be done in $\bigO(\size{F}^2)$ time.\qquad
\end{proof}

\begin{remark}\label{rem40}%
    The upper bound for the size of a small frame is achieved by the
    complete relation $D^n\!$.  We exhibit a small frame for $D^n$ in
    Lemma~\ref{lem130} below.  However, a frame can be much smaller
    than this upper bound $n(q-1)+1$.  Consider, for example, the
    $n$-ary relation $R=\set{(a,\dots,a): a\in D}$.  It is easy to
    show that $R$ is strongly rectangular.  However, it is also easy
    to see that $F = R$ is a frame, with $\wit(a,i) = (a,\dots,a)$
    ($i\in[n]$) and $|F|=q$.
\end{remark}

\begin{remark}
\label{rem50}%
    The \emph{compact representations} of Bulatov and
    Dalmau~\cite{BulDal06} are not necessarily frames and can have
    size $nq^2/2$.  However, it appears that a frame could be constructed efficiently from such a representation using methods
    similar to those of Lemma~\ref{lem70}.
\end{remark}

We will suppose below that all frames are small.  If necessary, this
can be achieved using Lemma~\ref{lem70}.  Note that we do not
assume that a frame for $R$ can actually generate $R$, since this is
entailed by the following.

\begin{lemma}
\label{lem80}
    If $R$ is strongly rectangular with \maltsev{} polymorphism $\ph$
    and $F$ is a frame for $R$, then $\cl{F}=R$.
\end{lemma}
\begin{proof}
    $F\subseteq R$ so $\cl{F} \subseteq \cl{R} = R$.  It remains to
    show that $R\subseteq \cl{F}$.

    We show by induction on $i\in[n]$ that $\proj[{[i]}] R\subseteq
    \proj[{[i]}]\cl{F}$.  The base case, $i=1$, is trivial as $\proj[1]
    R = \proj[1] F$ by definition.  Suppose that $\proj[{[i-1]}] R
    \subseteq \proj[{[i-1]}]\cl{F}$ and let $\bt = (t_1, \dots, t_n) =
    (\bu, t_i, \bv)\in R$.  By the inductive hypothesis, we have
    $\bu\in \proj[{[i-1]}] \cl{F}$ so there is a tuple $\bt' = (\bu,
    t'_i, \bv')\in \cl{F}\subseteq R$.  Therefore, $t'_i\tw t_i$,
    which means there are tuples $(\bu'\!, t_i, \bw)$ and $(\bu'\!, t'_i,
    \bw')$ in $F$ witnessing $(t_i, i)$ and $(t'_i, i)$, respectively.
    Thus, we have
    \begin{equation*}
        \begin{array}{c@{\hspace{1cm}}c@{\hspace{0.5cm}}c}
            \bu\phdash & t'_i & \bv'                     \\
            \bu'       & t'_i & \bw'                     \\
            \bu'       & t_i  & \bw\phdash               \\ \hline
            \bu\phdash & t_i  & \ph(\bv'\!, \bw'\!, \bw)\,.
        \end{array}
    \end{equation*}
    Therefore, $(t_1, \dots, t_i)\in \proj[{[i]}]\cl{F}$, continuing the
    induction.\qquad
\end{proof}

Given \ph and the matrix for $F$, the procedure of Lemma~\ref{lem80}
can be used to decide $\bt\in R$ in time $\bigO(n^2)$.  There is no need
to generate the whole of $R$; we just keep track of the tuple $(\bu,
t_i, \ph(\bv'\!, \bw'\!, \bw))$ that witnesses that $(t_1, \dots,
t_i)\in \proj[{[i]}] \cl{F}$.  If the procedure succeeds, we have
demonstrated that $\bt\in \cl{F} = R$; otherwise, we conclude either
that $t\notin R$ or that $R$ is not strongly rectangular.

We now show how, given a frame for $R$, we can determine a frame for
the relation
\begin{equation*}
   R(a_1, \dots, a_i, x_{i+1}, \dots, x_n)
        = \set{\bt\in R: (t_1, \dots, t_i) = (a_1, \dots, a_i)}\,.
\end{equation*}
\begin{lemma}
\label{lem90}
    Given a small frame $F$ for $R(x_1,x_2,\dots,x_n)$, a frame for
    $R(a,x_2,\dots,x_n)$ can be constructed in $\bigO(n^2)$ time.
\end{lemma}
\begin{proof}
    We abbreviate $R(a,x_2,\dots,x_n)$ to $R(a,\cdot)$.  For each
    $i=2,\dots,n$, determine $\cl{\proj[1,i]{F}} = \proj[1,i]{\cl{F}}
    = \proj[1,i]{R}$.  Note that $|\proj[1,i]{R}|\leq q^2\!$ and $\size{F} =
    \bigO(n)$ so this requires
    $\bigO(n)$ time for each $i$, and $\bigO(n^2)$ time in total.  We have
    $(a,b)\in \proj[1,i]{R}$ if, and only if, $b\in\proj{R(a,
    \cdot)}$.  Also, we have calculated a witness (with respect to
    $R$) for each $b\in
    \proj{R(a,\cdot)}$.  Let \tw be the usual congruence for $R$ and
    $\tw'$ the corresponding congruence for $R(a,\cdot)$.  Clearly
    $b\tw' c$ implies $b\tw c$, since there are witnesses
    $(a,\bu,b,\bv), (a,\bu,c,\bv')\in R$.  On the other hand, if $b\tw
    c$ and $b\in\proj{R(a,\cdot)}$, then $c\in \proj{R(a, \cdot)}$ and
    $b\tw' c$, since we have
    \begin{equation*}
        \begin{array}{c@{\hspace{1cm}}c@{\hspace{1cm}}c@{\hspace{5mm}}c}
            a\phdash  & \bu\phdash  & b & \bv\phantom{{}''}    \\
            a'        & \bu'        & b & \bv'\phdash          \\
            a'        & \bu'        & c & \bv''                \\ \hline
            a\phdash  & \bu\phdash  & c & \ph(\bv,\bv'\!,\bv'')\,.
        \end{array}
    \end{equation*}
    Thus, the equivalence classes of $\tw'$ are a subset of those of
    \tw.  Therefore we can construct $\tw'$ and a witness for
    each $b\in\proj{R(a,\cdot)}$, using $F$ and the $n$-tuples from
    the calculation of $\proj[1,i]{R}$.\qquad
\end{proof}

The following corollary is immediate, by iterating the
Lemma~\ref{lem90} $i\leq n$ times.

\begin{corollary}
\label{cor100}
    Given a frame $F$ for the relation $R(x_1,x_2,\dots,x_n)$, we can construct a frame for
    $R(a_1,\dots, a_i,x_{i+1},\dots,x_n)$ in
    $\bigO(n^3)$ time.
\end{corollary}


\section{Constructing a frame}
\label{sec:Frame}

If $R$ is $\Gamma$-definable, then $\bt\in R$ can be decided in
polynomial time by checking that $\bt$ satisfies each of the defining
constraints.  We cannot use this method to decide $R =
\emptyset$ efficiently but this can be done trivially using any frame $F$ for $R$,
since $R=\emptyset$ if, and only if, $F=\emptyset$.  If $F\neq
\emptyset$, then any $\Bf\in F$ is a certificate that $R\neq
\emptyset$.  Similarly, given a frame for $R$ and any tuple
$(a_1,\dots,a_i)$, we can determine whether there is any $\bt\in R$
such that $(t_1,\dots,t_i)=(a_1,\dots,a_i)$, using the method of
Corollary~\ref{cor100}.

However, we must be able to construct some frame $F$ for $R$
efficiently.  If $\Gamma$ is strongly rectangular, we will show how to
determine a frame for a $\Gamma$-formula $\Phi$ having $m$ constraints
in $n$ variables, in time polynomial in $m$, $n$ and $\size{\Gamma}$.
This is achieved, as in~\cite{BulDal06}, by adding the constraints
sequentially.

If the $m$ constraints are $\con_1,\con_2,\dots,\con_m$, let $\Phi_s =
\con_1\wedge\con_2\wedge\cdots\wedge\con_s$.  Thus, $\Phi_0 = D^n\!$, the
complete $n$-ary relation on $D$, and $\Phi_m=\Phi$.  We begin by
constructing a frame for $\Phi_0$.

\begin{lemma}
\label{lem130}
    A small frame $F_0$ for $\Phi_0$ can be constructed in $\bigO(n)$
    time.
\end{lemma}
\begin{proof}
    Let $d$ be any element of $D$ and let $F_0 = \set{\bt^d} \cup
    \set{\bt^{a,i}: i\in[n], a\in D\setminus d}$, where
    \begin{equation*}
       t^d_j = d
        \quad\text{and}\quad
        t^{a,i}_j = \begin{cases}\ a &\text{ if } j = i\\
                                \ d &\text{ otherwise}
                   \end{cases}
        \qquad(j\in[n]).
    \end{equation*}
    Clearly all these tuples are in $\Phi_0$.  Also $\wit(d,i)=\bt^d$
    and $\wit(a,i)=\bt^{a,i}$ ($a\neq d$), for all $i\in[n]$, is a
    witness function.  Further, we have $\proj[{[i-1]}]{\bt^{a,i}} =
    \proj[{[i-1]}]{\bt^{d}} = (d,\dots,d)$.  Thus, $F_0$ satisfies the
    conditions for being a frame.  We have $|F_0| = n(q-1)+1$, so
    $F_0$ is small.\qquad
\end{proof}

Note that $|F_0|$ matches the upper bound for the size of a small
frame.

Now, we show how to determine a frame for $\Phi_s$ given a frame for
$\Phi_{s-1}$.  We first show that this can be done in polynomial time
when $\size{\Gamma}=\bigO(1)$.  This is nonuniform \csp, the most
important case.

\begin{lemma}
\label{lem140}
    Given a frame $F$ for $\Phi$ and a constraint $\con$, a frame $F'$
    for $\Phi'=\Phi\wedge\con$ can be constructed in $\bigO(n^4)$ time.
\end{lemma}
\begin{proof}
    Suppose that $\con = H(x_{i_1},x_{i_2},\dots,x_{i_r})$, where
    $H\in\Gamma$ has arity $r$.  We will assume that $x_{i_1},
    x_{i_2}, \dots, x_{i_r}$ are distinct since, otherwise, we can
    consider a smaller relation $H'$ over the distinct variables.  Let
    $I = \set{i_1,i_2,\dots,i_r}$.  For each $i\in [n]$, let $J_i =
    I\cup\set{i}$ and determine $T_i\subseteq \Phi$ such that
    $\proj[J_i] T_i = \cl{\proj[J_i]{\Phi}}$ using \closure.  If $\ell =
    |\proj[I]{\Phi}|$, then $|T_i|\leq q\ell$, so this takes time
    $\bigO(n\ell^3 + r\ell^4)$ by Lemma~\ref{lem60}.  But, since
    $\size{\Gamma} = \bigO(1)$, we have $r=\bigO(1)$, $\ell\leq q^r=\bigO(1)$ and
    $\bigO(n\ell^3 + r\ell^4) = \bigO(n)$.  The entire computation for all $i$
    therefore takes time $\bigO(n^2)$ and we have $\sum_i|T_i|=\bigO(n)$.

    Determine $U_i$, the set of tuples in $T_i$ that are consistent
    with $\Theta$, so $U_i \subseteq \Phi'\!$.  Now $U_i$ contains a
    witness for each $a \in \proj{\Phi'}\!$, since
    \begin{equation*}
       \proj[J_i]{U_i}
            = (\proj[J_i] T_i) \cap \con
            = (\cl{\proj[J_i]{F}})\cap \con
            = (\proj[J_i]{\Phi})\cap \con
            = \proj[J_i]{(\Phi\wedge \con)}
            = \proj[J_i]{\Phi'}\,.
    \end{equation*}
    Thus, in particular, $\proj{U_i} = \proj{\Phi'}\!$.  We now do the
    following for each $i\in[n]$.

    Let $\cA \gets \proj{U_i}$ and repeat the following until $\cA =
    \emptyset$.  Choose $\bt\in U_i$ such that $t_i\in \cA$.
    Determine a frame $F^\star$ for $\Phi(t_1, \dots, t_{i-1}, x_i,
    \dots, x_n)$ in $\bigO(n^3)$ time, using Corollary~\ref{cor100}. Clearly
    $\bt\in \cl{F^\star}\!$, so $F^\star \neq \emptyset$. Now
determine the intersection of $\Theta$ with the relation $R^\star =
\Phi(t_1, \dots, t_{i-1}, x_i, \dots, x_n)$ generated by $F^\star\!$, using \closure, as was done for $\Phi$ above. This takes $\bigO(n)$ time; let the resulting relation be $R^\circ\!$. Now, by Corollary~\ref{cor30}, $\proj{R^\circ}$ is the equivalence class $\cE = \set{a: a\tw' t_i}$ of $t_i$ in $\Phi'\!$.  For each $a\in \cE$, we can find a witness $\wit'(a,i)\in R^\circ$ for $a \in \proj{\Phi'}$ and these have the common prefix $(t_1, \dots, t_{i-1})$. We set $\cA\gets \cA\setminus \cE$, and repeat.

    At the end of this process, $\wit'$ is the witness
    function for a frame $F'$ for $\Phi'\!$.  The total time required is
    $\bigO(n^3|F'|) = \bigO(n^4)$.\qquad
\end{proof}

\begin{lemma}
\label{lem150}
    A frame $F$ for $\Phi$ can be constructed in time $\bigO(mn^4)$.
\end{lemma}
\begin{proof}
    Construct $\Phi_0$ in $\bigO(n)$ time.  Then, apply
    Lemma~\ref{lem140} to construct a frame $F_i$ for $\Phi_i$ from a
    frame $F_{i-1}$ for $\Phi_{i-1}$, for each $i\in[m]$.  At
    termination, set $\Phi\gets\Phi_m$ and $F\gets F_m$.\qquad
\end{proof}

Since a relation has $\emptyset$ for a frame if, and only if, it is
empty (and $\emptyset$ has no other frame), we can determine in time
$\bigO(mn^4)$ whether there is a satisfying assignment to a \csp{}
instance in a fixed strongly rectangular vocabulary.  By
Lemma~\ref{lem20}, we have re-proven the main result of~\cite{BulDal06}.

We assumed above that $\size{\Gamma} = \bigO(1)$.  However, we can still
perform the computations of Lemma~\ref{lem140} in time polynomial in
$m$, $n$ and $\size{\Gamma}$.

\begin{lemma}
\label{lem160}
    A frame for $\Phi$ can be constructed in time $\bigO(mn^4 +
    mn^2\size{\Gamma}^4)$.
\end{lemma}
\begin{proof}
    We indicate how the proof of Lemma~\ref{lem140} must be modified.
    It is only the computation of the $U_i$ that requires improvement,
    which we achieve by using a device from~\cite{BulDal06}.  Suppose
    we wish to add a constraint $\con = H(x_{i_1}, x_{i_2}, \dots,
    x_{i_r})$ to $\Phi$.  Instead, we add in turn the $r$ constraints
    $\con_k = H_k(x_{i_1}, x_{i_2}, \dots, x_{i_k})$, where $H_k =
    \proj[{[k]}]H$ for each $k\in[r]$.  Thus, $|H_1|\leq q$ and $H_r =
    H$.  Letting $\Psi_0 = \Phi$, we successively calculate frames for
    $\Psi_k = \Psi_{k-1} \wedge \con_k$ ($k\in[r]$), so $\Psi_r =
    \Phi'\!$.

    If $I_k = \set{i_1,i_2,\dots,i_k}$ ($k\in[r]$), we have
    \begin{equation*}
       \ell_k \ = \ |\proj[I_k]\Psi_{k-1}|
            \ \leq \ q|\proj[I_{k-1}]\Psi_{k-1}|
            \ \leq \ q|H_{k-1}| \ \leq \ q|H|\,.
     \end{equation*}
    Thus, for each $k\in[r]$, the time required to compute $U_i$ and $R^\circ$ in
    Lemma~\ref{lem140} becomes $\bigO(n^2|H|^3 + nr|H|^4)$.  In total, the time requirement
    is $\bigO(n^2r|H|^3 + nr^2|H|^4) = \bigO(n^2\size{H}^4) =
    \bigO(n^2\size{\Gamma}^4)$.\qquad
\end{proof}


\section{Counting problems}
\label{sec:Counting}

We consider the problem of determining $|R_\Phi|$, which we abbreviate
to $|\Phi|$, where $\Phi$ is a $\Gamma$-formula with $m$ constraints
and $n$ variables.  We require the computations to be done in time
polynomial in the size of the input $\Phi$ and we assume
$\size{\Gamma} = \bigO(1)$.  In fact, the size of $\Phi$ can be measured
by a polynomial in $n$.  A repeat of a constraint can be removed, since this
does not change $R_\Phi$.  Then an $r$-ary relation in $\Gamma$ can
give rise to $\bigO(n^r)$ constraints. We will assume that every
variable appears in at least one constraint. Otherwise, suppose $n_0$
variables do not appear: letting $\Phi'$ be $\Phi$  with these
variables deleted, we have $|\Phi|=q^{n_0}|\Phi'|$. Hence we will
assume that $m = \Omega(n)$.

Following Bulatov and Dalmau~\cite{BulDal07}, we call this
computational problem \ncsp.  If $\Gamma = \set{H,\equ}$, we write
\ncsp[H]. We will use the following result from~\cite{BulDal07},
which we prove here for completeness. The corollary is immediate.
\begin{theorem}[Bulatov and Dalmau~\cite{BulDal07}]
\label{thm10}
    Let $\frS=(D,\Gamma)$, $\frS'=(D,\Gamma')$ be relational structures
    with $\Gamma'\subseteq\cocl[\Gamma]$.  Then \ncsp[\Gamma'] is
    polynomial-time reducible to \ncsp.
\end{theorem}
\begin{proof}
Let each $H'\in\Gamma'$ have pp-definition $H'(\bx)=\exists
\by\,H^*(\bx,\by)$, with $H^*(\bx,\by)$ a $\Gamma$-formula. If all
relations in $\Gamma$ have arity at most $r$ and at most $\ell$ tuples
and all the formulae $H^*$ are conjunctions of at most $k$
constraints, then each $H^*$ has arity at most $kr$ and $|H^*|\leq
\ell^k\!$. Observe that $k$, $\ell$ and $r$ are constants in
\ncsp[\Gamma'].

Consider any $\Gamma'$ formula $\Phi(\bx)=\Theta_1\wedge\cdots\wedge\Theta_m$,
where $\bx=(x_1,\ldots,x_n)$. Now, if $\Theta_i=H'(\bx)$, let $\Theta^*_i = H^*(\bx,\by_i)$, where the $\by_i$ ($i\in[m]$) are new variables. Let $\bz=(\by_1,\ldots,\by_m)$ and consider the $\Gamma$-formula $\Phi^*(\bx,\bz) = \Theta^*_1\wedge\cdots\wedge\Theta^*_m$. This is an instance of \ncsp, with at most $km$ constraints and $n+krm$ variables. Now, for $\bx\in\Phi$, let
\begin{equation*}N_i(\bx)\ =\ \big|\set{\by_i\,\colon(\bx,\by_i)\in \Theta^*_i}\big|\ \leq\ |H^*|\ \leq \ell^k\qquad(i\in[m]),
\end{equation*}
and let $N=\max\set{N_i(\bx): i\in[m],\,\bx\in\Phi}\leq \ell^k\!$. Now let
\begin{equation*}
   \mu_j(\bx)\ =\ \big|\set{i\in[m]\,\colon N_i(\bx)=j}\big| \qquad(j\in[N]).
\end{equation*}
Clearly $\sum_{j=1}^N \mu_j(\bx)=m$ for all $\bx\in \Phi$. Let
\begin{equation*}
    \bM\ =\ \set{(\mu_1(\bx), \dots, \mu_N(\bx)): \bx\in \Phi}\,.
\end{equation*}
Let $L=|\bM|$. Clearly, $|\bM|< m^N\!$, so $L$ has bit-size $\bigO(m)$. Now, for $\bm\in\bM$, let
\begin{equation*}
K(\bm)\ =\ \big|\set{\bx\in \Phi\,\colon \mu_j(\bx)=m_j,\, j\in[N]}\big|\ \leq\ q^n\ \leq\ q^m\,.
\end{equation*}
Thus, $|\Phi|=\sum_{\bm\in\bM}K(\bm)$.
Now let $J(\bm)=\prod_{j=1}^N {j}^{m_j}< N^m\!$. Thus, the $J(\bm)$, $K(\bm)$ ($\bm\in[\bM]$) are numbers with $\bigO(m)$ bits. Then we have
\begin{equation*}
|\Phi^*|\ =\ \sum_{\bx\in \Phi}\, \prod_{i\in[m]} N_i(\bx)
\ =\ \sum_{\bm\in\bM} K(\bm) \prod_{j=1}^N {j}^{m_j}
\ =\ \sum_{\bm\in\bM} K(\bm) J(\bm)\,.
\end{equation*}

Now, for $s\in[L]$, consider the  $\Gamma$-formulae
\begin{equation*}
\Phi^*_s(\bx,\bz_1,\ldots,\bz_s)\ =\ \bigwedge_{i\in[s]}\Phi^*(\bx,\bz_i)\,,
\end{equation*}
where $\bz_i$ ($i\in[s]$) are distinct variables.
Then $\Phi^*_s$ is an instance of \ncsp, with at most $kms$ constraints and $krms$ variables, and we clearly have
\begin{equation*}
|\Phi^*_s|\ =\ \sum_{\bm\in\bM} K(\bm) J(\bm)^s.
\end{equation*}
Note that $\Phi^*_s$ is of size polynomial in $m$. Therefore we can evaluate $|\Phi^*_s|$ for all $s\in[L]$ using a polynomial number of calls to an oracle for \ncsp, each having input of size polynomial in $m$. It then follows, using~\cite[Lemma~3.2]{DyeGre00}, that we can recover  $\sum_{\bm\in\bM}K(\bm)=|\Phi|$ from the values of the $|\Phi^*_s|$ ($s\in[L]$) in time polynomial in $L$, which is polynomial in $m$.\qquad
\end{proof}
\begin{corollary}
\label{cor40}
    If $H\in\cocl$ and \ncsp[H] is \numpc, then \ncsp is \numpc.
\end{corollary}

First, we apply Corollary~\ref{cor40} to give a short proof of the
main result of~\cite{BulDal07}.  (Bulatov and Dalmau phrase the result
in terms of the existence of a \maltsev polymorphism but, by
Lemma~\ref{lem20}, our phrasing is equivalent.)

\begin{lemma}[Bulatov and Dalmau~\cite{BulDal07}]
\label{lem170}
    If the constraint language $\Gamma$ is not strongly rectangular, then \ncsp is \numpc.
\end{lemma}
\begin{proof}
    Clearly $\ncsp\in \nump$ for any $\Gamma\!$.  If $\Gamma$ is not
    strongly rectangular, there is an $r$-ary relation $B\in\cocl$
    that is not rectangular when considered as a binary relation over
    $D^k\times D^{r-k}$ for some $k$ with $1\leq k<r$.  Let $G =
    (V,E)$ be a connected, undirected bipartite graph with vertex
    bipartition $V_1,V_2$.  Let $\Phi_1$ be the $\Gamma$-formula with
    a constraint $B(\bx_i,\bx_j)$ for each $\set{\nu_i, \nu_j}\in E$
    with $\nu_i\in V_1$, $\nu_j\in V_2$.  Define $\Phi_2$ analogously,
    but with constraints $B(\bx_j,\bx_i)$.  It follows that $|\Phi_1|
    + |\Phi_2|$ is the number of graph homomorphisms from $G$ to
    $\cG_B$.  This problem is \numpc by~\cite{DyeGre00}, since $\cG_B$
    has a component which is not a bipartite clique.  Thus, \ncsp[B]
    is \numpc and, hence, \ncsp is \numpc by Corollary~\ref{cor40}.\qquad
\end{proof}

There is an important generalisation of the counting problem to
\emph{weighted} problems which we now describe briefly;
see~\cite{BulGro05,DyGoJe09} for details.  The relations $H\subseteq
D^r$ in $\Gamma$ are replaced by functions $f\colon D^r\to \nnrats\!$,
where $\nnrats$ denotes the non-negative rationals.\footnote{More
generally, we can take the function values to be non-negative
algebraic numbers.} Thus, $\Gamma$ is replaced by a set of
functions $\cF$.  We will call $(D,\cF)$ a \emph{weighted structure}.
The \emph{underlying relation} of $f\in\cF$ is $\set{\bu\in D^r :
f(\bu)>0}$.  Note that a relation $H$ can be identified with a
function $f_H\colon D^r\to\set{0,1}$, where $f_H(\bu)=1$ if, and only
if, $\bu\in H$.  Then $H$ is the underlying relation of $f_H$.  Thus, we
may just use $H$ to denote the function $f_H$ without further
comment.

Now, using notation similar to the relational case, an instance $\cI$
of \ncsp[\cF] is defined as follows.  A constraint $\Theta$ has the
form $f(x_{i_1}, x_{i_2}, \dots, x_{i_r})$ for some $r$-ary function
$f\in\cF$.  Thus, $(\nu_{i_1}, \nu_{i_2}, \dots, \nu_{i_r})$ is the
scope of the $\Theta$.  Suppose we have constraints $\Theta_1, \dots,
\Theta_m$, where $\Theta_s$ applies the function $f_s\in\cF$.  Write
$\bx_s$ for $(x_{i_1}, x_{i_2}, \dots, x_{i_r})$, where $(\nu_{i_1},
\nu_{i_2}, \dots, \nu_{i_r})$ is the scope of the $\Theta_s$.  Then,
the \emph{weight} of an assignment $\bx\colon V\to D$ is
\begin{equation*}
   \sW(\bx) = \prod_{s=1}^m f_s(\bx_s)\,.
\end{equation*}
The computational problem \ncsp[\cF] is then to compute the
\emph{partition function},
\begin{equation*}
   Z(\cI) = \!\!\sum_{\bx\colon V\to D}\!\! \sW(\bx)\,.
\end{equation*}
If $\cF=\set{f}$ for a single function $f\!$, we write \ncsp[f].

We may view a binary function $f\colon A_1\times A_2\to \nnrats$ as a
matrix with elements in $\nnrats\!$, rows indexed by $A_1$ and columns
indexed by $A_2$.  If $B$ is its underlying relation,
the submatrix of $f$ induced by a block of $B$ is called a block of $f\!$. If $f_1,f_2,\ldots,f_k$ are the blocks of $f\!$, then $f$ will be called a \emph{rank-one block matrix}, if each block of $f$ is a rank one matrix.

\begin{lemma}
\label{lemma:rank-one-rectangular}
    If $f\colon A_1\times A_2\to\nnrats$ is a rank-one block matrix,
    its underlying relation $B$ is rectangular.
\end{lemma}
\begin{proof}
    If $B$ is not rectangular, there are $(a,c),\,(b,c),\,(a,d)\in R$
    such that $(b,d)\notin B$.  The $2\times 2$ sub-matrix of $f$
    induced by rows $a,\,b$ and columns $c,\,d$ is included within a
    single block and has determinant $-f(a,d)f(b,c)\neq 0$ and so has
    rank~2.  Therefore, the block of $f$ that contains this sub-matrix
    has rank at least~2.
\end{proof}

We will call a matrix $f\colon A_1\times A_2\to\nnrats$ \emph{rectangular} if its underlying relation $R$ is rectangular.
Thus, an alternative way of defining a rank-one block matrix is as
a rectangular matrix $f\!$, together with functions $\alpha_1\colon A_1\to\nnrats\!$, $\alpha_2\colon A_2\to\nnrats\!$, such that $f(x,y) = \alpha_1(x) \alpha_2(y)$ for all $(x,y)\in B$.

We can now state a theorem of Bulatov and
Grohe~\cite[Theorem~14]{BulGro05}, which generalises the result of Dyer
and Greenhill~\cite{DyeGre00} to the weighted case.  Although we give
the theorem for non-negative rational functions, in fact we only
require the case for non-negative integer functions.

\begin{theorem}[Bulatov and Grohe~\cite{BulGro05}]
\label{thm30}
    Let $f\colon A_1\times A_2\to \nnrats$ be a binary function.  Then
    \ncsp[f] is in \fp if $f$ is a rank-one block matrix.  Otherwise \ncsp[f] is \numph.
\end{theorem}

In Section~\ref{sec:Dichotomy:Alg}, we will use the following property of rank-one block matrices.
\begin{lemma}\label{lem173}
    If $f\colon A_1\times A_2\to\nnrats$ is a rank-one block matrix,
    it is uniquely determined by its underlying relation and its row and
    column totals.
\end{lemma}
\begin{proof}
    Let $B$ be the underlying (rectangular) relation.  Consider any
    block $C$ of $B$, with $\proj[1]{C} = S_1$, $\proj[2]{C} = S_2$.
    Then there exist $\alpha_1\colon S_1\to\nnrats$ and
    $\alpha_2\colon S_2\to\nnrats$ such that $f(x_1,x_2) =
    \alpha_1(x_1) \alpha_2(x_2)$ for every $x_1\in S_1$ and $x_2\in
    S_2$.  Now, let
        \begin{alignat*}{3}
        & &f(x_1, \cdot)\
            &=\ \sum_{x_2\in S_2}f(x_1,x_2)\ &
            &=\ \alpha_1(x_1) \sum_{x_2\in S_2}\alpha_2(x_2) \\
        & &f(\cdot, x_2)\
            &=\ \sum_{x_1\in S_1} f(x_1,x_2)\ &
            &=\ \alpha_2(x_2) \sum_{x_1\in S_1}\alpha_1(x_1) \\
        & &f(\cdot, \cdot)\
            &=\ \sum_{x_1\in S_1} f(x_1, \cdot)\ &
            &=\ \sum_{x_1\in S_1} \alpha_1(x_1) \sum_{x_2\in S_2}
                \alpha_2(x_2)
        \end{alignat*}
    be the row, column and grand totals of $f(x_1,x_2)$ ($x_1\in S_1,
    x_2\in S_2$).  A simple calculation gives
    \begin{equation*}
       f(x_1,x_2)\ =\ \frac{f(x_1 , \cdot) f(\cdot, x_2)}
                          {f(\cdot, \cdot)}\,. \qedhere
    \end{equation*}
\end{proof}


\section{The dichotomy theorem}
\label{sec:Dichotomy}

We are now ready to describe the dichotomy.  We saw in the previous
section that, assuming $\fp\neq \nump$, strong rectangularity
is a necessary condition for tractability.  In this section, we
introduce a stronger condition, based on certain rank-one block
matrices and show that it characterises the  dichotomy for \numcsp{},
into problems in \fp{} and problems which are \numpc{}.
As one would expect, this condition turns
out to be equivalent to the criterion in Bulatov's dichotomy theorem.
We defer the algorithm for the polynomial-time cases to
Section~\ref{sec:Dichotomy:Alg} and some technical results to
Section~\ref{sec:Dichotomy:Cong}. In Section~\ref{sec:Decide},
we will show that the condition is decidable.

Let $H(x,y,z)$ be a ternary relation on $A_1\times A_2\times A_3$.  We will
call $H$ \emph{balanced} if the \emph{balance matrix}, 
\begin{equation*}
   M(x,y) = |\set{z\in A_3: (x,y,z)\in H}| \qquad (x\in A_1,\ y\in A_2)
\end{equation*}
is a rank-one block matrix.  A relation of arity $n>3$ is balanced if
every expression of it as a ternary relation on $D^k\times
D^\ell\times D^{n-k-\ell}$ ($d,\ell\geq 1$, $k+\ell<n$) is balanced.
We will say that $\Gamma$ is \emph{strongly balanced} if every
pp-definable ternary relation is balanced.

We will prove the following dichotomy theorem.

\begin{theorem}
\label{thm40}
    If $\Gamma$ is strongly balanced, \ncsp is in \fp.  Otherwise,
    \ncsp is \numpc. Moreover, the dichotomy is decidable.
\end{theorem}
\begin{proof}
    The first statement will be proved in
    Section~\ref{sec:Dichotomy:Alg}.  The second is proved in
    Lemma~\ref{lem180} below. The third is proved in Section~\ref{sec:Decide}.\qquad
\end{proof}

We first show that the condition of strong balance is strictly
stronger than that of strong rectangularity.

\begin{lemma}
\label{lem175}
    Strong balance implies strong rectangularity.
\end{lemma}
\begin{proof}
    This follows from the definition of strong
    balance.  Suppose $\Gamma$ is strongly balanced and let $B(x,y)$
    be any definable binary relation.  Let
    \begin{equation*}
        H(x,y,z)=\exists w\, B(x,y)\wedge B(z,w)\,,
    \end{equation*}
    which must be balanced.  Then $M(x,y) = |\set{z:
    \exists w\, B(z,w)}|= |\proj[1]{B}|$, for all
    $(x,y)\in B$.  If $|\proj[1]{B}| = 0$
    then $B = \emptyset$, which is trivially rectangular.  Otherwise,
    the underlying relation of $M$ is $B$, which must be rectangular
    by Lemma~\ref{lemma:rank-one-rectangular}.\qquad
\end{proof}

The converse of Lemma~\ref{lem175} is not true, however.

\begin{lemma}
\label{lem270}
    Strong rectangularity does not imply strong balance.
\end{lemma}
\begin{proof}
    Consider the following example.  Let $A = \set{a_{0,0}, a_{0,1},
    a_{1,0}, a_{1,1}, b}$ and let $D = A \cup \set{0,1}$.  Let $\Gamma
    = \set{R}$, where $R$ is the ternary relation given by
    \begin{equation*}
        R = \set{(i, j, a_{i,j}) : i,j\in\set{0,1}} \cup \set{(0,0,b)}\,.
    \end{equation*}
    Note that $b$ is, in effect, a second copy of $a_{0,0}$; the
    effect is essentially that of a weighted relation where the tuple
    $(0,0,a_{0,0})$ has weight~2 and all other tuples have unit weight.
    The balance matrix $M$ for $R$ is as follows (we omit the rows and
    columns for $x\in A$ as they have only zeroes):
    \begin{equation*}%
        \raisebox{-1.5ex}{$M\ =\ $}
        \begin{array}{r}  \\ 0\\ 1\end{array}\hspace{-9pt}
        \begin{array}{c}
            \begin{array}{cc}  0 & 1 \end{array}\\[-1pt]
            \left[\begin{array}{ccc}
                2 & 1 \\
                1 & 1 \\
            \end{array}\right]\,.
        \end{array}
    \end{equation*}

    $M$ is clearly not a rank-1 block matrix, so $R$ is not strongly
    balanced.  Nonetheless, we will show that $R$ has a \maltsev
    polymorphism.  Consider the following function, where $\oplus$
    denotes addition modulo~2.
    \begin{equation*}
        f(x,y,z) = \begin{cases}
            \ x\oplus y\oplus z &
                 \text{if } x,y,z\in\set{0,1}                         \\
            \ a_{f(i, k, m), f(j, \ell, n)} &
                 \text{if } x=a_{i,j}, y=a_{k,\ell}, z=a_{m,n}        \\
            \ a_{0,0} &
                 \text{otherwise.}
        \end{cases}
    \end{equation*}
    Let $g(b) = a_{0,0}$ and $g(x) = x$ for all other $x\in D$.  We
    define the function $\ph$ as follows:
    \begin{equation*}
        \ph(x,y,z) = \begin{cases}
            \ x                   &\text{if }y=z                      \\
            \ z                   &\text{if }x=y                      \\
            \ f(g(x), g(y), g(z)) &\text{otherwise.}
        \end{cases}
    \end{equation*}

    In other words, $\ph$ behaves identically to $f\!$, except that it
    has the \maltsev property and, for inputs where $x\neq y$ and
    $y\neq z$, it ``pretends'' that any input of $b$ is actually an
    input of $a_{0,0}$.  Note that, for $i,j,k\in \set{0,1}$,
    $\ph(i,j,k) = i\oplus j\oplus k$, regardless of the \maltsev
    condition.

    We claim that, as well as being \maltsev, $\ph$ is a polymorphism
    of $R$.  To this end, let $\bx,\by,\bz\in R$, which we can write
    as $\bx=(i,j,x')$, $\by=(k,\ell,y')$ and $\bz=(m,n,z')$, where
    $x'=a_{i,j}$ or, if $i=j=0$, we may have $x'=b$, and similarly for
    $y'$ and $z'\!$.  So, we have
    \begin{align*}
        \ph(\bx,\by,\bz)
            &= \big(\ph(i,k,m),\ph(j,\ell,n),\ph(x'\!, y'\!, z')\big) \\
            &= \big(f(i,k,m), f(j,\ell,n), f(g(x'), g(y'), g(z')\big) \\
            &= \big(f(i,k,m), f(j,\ell,n),
                                       a_{f(i,k,m), f(j,\ell,n)}\big) \\
            &\in R\,.
    \end{align*}
    This establishes the claim.\qquad
\end{proof}

\begin{remark}
\label{rem105}%
    The example in Lemma~\ref{lem270} can be extended to relations of
    arbitrary size by extending $i$ and $j$ in the tuples
    $(i,j,a_{i,j})$ to longer binary strings and interpreting $\oplus$
    as bit-wise XOR (e.g., $0011\oplus 0101 = 0110$).
\end{remark}

\begin{remark}\label{rem110}%
    Bulatov and Dalmau conjectured in~\cite{BulDal03} that a \maltsev polymorphism was sufficient for \ncsp to be in \fp. That is a stronger claim than the converse of Lemma~\ref{lem175}. The conjecture was withdrawn in~\cite{BulDal07}, with a counterexample somewhat similar to that in the proof of Lemma~\ref{lem270}.
\end{remark}

Next, we strengthen Lemma~\ref{lem170} to prove one half of the
dichotomy.

\begin{lemma}
\label{lem180}
    If $\Gamma$ is not strongly balanced, then \ncsp is \numpc.
\end{lemma}
\begin{proof}
    If $\Gamma$ is not strongly balanced, there is an unbalanced
    ternary relation $H\in\cocl$.  Let $E$ be a binary relation with
    $V = V_1\cup V_2$, $V_1\cap V_2 = \emptyset$ and $\proj{E} =
    V_i$ ($i=1,2$).  Let $\Phi$ be the $\Gamma$-formula with a
    constraint $H(x_i,x_j,z_{ij})$ for each $(\nu_i,\nu_j)\in E$.
    Thus, $\Phi$ has $|V| + |E|$ variables and $|E|$ constraints.  Let
    $M\colon V_1\times V_2\to\nnrats$ be $\Phi$'s balance matrix.

    We have $|\Phi| = Z(\cI)$, where $Z(\cI)$ is the partition
    function for an instance $\cI$ of \ncsp[M] with input $E$.  But
    this problem is \numph by Theorem~\ref{thm30} and, hence, \ncsp[H]
    is \numpc.  Thus, \ncsp is \numpc by Corollary~\ref{cor40}.\qquad
\end{proof}

In~\cite{Bulato08}, Bulatov defined \emph{congruence
    singularity}.  Suppose $\Gamma$ is a constraint language and
$\rho_1$ and $\rho_2$ are two congruences defined on the same pp-definable set
$A\subseteq D^r\!$.  Let the equivalence classes of $\rho_i$ be
$E_{ij}$ ($j\in[\nu_i]$, $i=1,2$).  Further, let
\begin{equation}
\label{eq20}
    \cM(j,k) = |E_{1j}\cap E_{2k}|\qquad(j\in [\nu_1],\ k\in [\nu_2]).
\end{equation}
$\Gamma$ is \emph{congruence singular} if $\cM$ is a rank-one
block matrix for every pair $\rho_1$, $\rho_2$ of congruences.\footnote{In
fact, Bulatov applies this term to the associated algebra, but
with essentially this meaning.}

\begin{lemma}
\label{lem190}
    $\Gamma$ is congruence singular if, and only if, it is strongly
    balanced.
\end{lemma}
\begin{proof}
    Suppose $\Gamma$ is strongly balanced, let $A\subseteq D^r$ be
    defined by the formula $\chi$ and let $\rho_1, \rho_2 \in \cocl$
    be congruences defined on $A\subseteq D^r$ with equivalence
    classes $E_{ij}$ ($j\in[\nu_i]$, $i=1,2$).  Then $\psi(\bx, \by,
    \bz) = \chi(\bz)\wedge \rho_1(\bx, \bz)\wedge \rho_2(\bz, \by)$ is
    a ternary relation.  Hence, for any $\bx\in E_{1j}$ and $\by\in
    E_{2k}$, the matrix
    \begin{equation*}
       M(\bx,\by)
            = |\set{\bz: \chi(\bz)\wedge \rho_1(\bx,\bz)\wedge \rho_2(\bz,\by)}|
            = |E_{1j}\cap E_{2k}|
    \end{equation*}
    is a rank-one block matrix.  But $M$ has a set of identical rows
    for all $\bx\in E_{1j}$ ($j\in[\nu_1]$) and a set of identical
    columns for all $\by\in E_{2k}$ ($k\in[\nu_2]$).  The matrix $\cM$
    has one representative from each of these sets.  It follows that
    $\cM$ is a rank-one block matrix.

    Now, suppose that $\Gamma$ is congruence singular and let $H\in
    \cocl$ be any ternary relation.  Define relations $\rho_i =
    \set{(\bx, \by): \bx, \by\in H \text{ and } x_i = y_i}$
    ($i=1,2$).  These are trivially equivalence relations, and are
    pp-definable as $H(x_1,x_2,x_3)\wedge H(y_1,y_2,y_3)\wedge
    (x_i=y_i)$.  Thus, they are two congruences defined on the same
    set, $H$, which is also pp-definable.  The equivalence classes of
    $\rho_i$ clearly correspond
    to $z_i\in\proj{H}$ ($i=1,2$) and we may index these classes by
    $z_i$.  Thus,
    \begin{align*}
        \cM(z_1,z_2)
            &= |\set{(x_1,x_2,x_3)\in H: x_1=z_1,x_2=z_2}|          \\
            &= |\set{x_3: (z_1,z_2,x_3)\in H}|                      \\
            &= M(z_1,z_2)\,.
    \end{align*}
    Since $\cM$ is a rank-one block matrix by assumption, so is $M$,
    and the conclusion follows.\qquad
\end{proof}

In~\cite{Bulato08}, Bulatov established the following theorem, giving a dichotomy for \numcsp that is equivalent, using~Lemma~\ref{lem190}, to Theorem~\ref{thm40}, except that the decidability of the dichotomy remained open.

\begin{theorem}[Bulatov~\cite{Bulato08}]
\label{thm50}
    If $\Gamma$ is congruence singular, \ncsp is in \fp.  Otherwise
    \ncsp is \numpc.
\end{theorem}


\subsection{The counting algorithm}
\label{sec:Dichotomy:Alg}

This section is devoted to a proof of the poly\-nomial-time case of the
dichotomy theorem.

\begin{lemma}
\label{lem195}
    Let $\Gamma$ be strongly balanced and let $R\in\cocl$ be an
    $n$-ary relation.  Given a frame $F$ for $R$, $|R|$ can be
    computed in $\bigO(n^5)$ time.
\end{lemma}
\begin{proof}
    If $n=1$ then $R = \proj[1] R = \proj[1] F = F$ so $|R| = |F|$ and
    we are done.  So we may assume that $n\geq 2$.  Now, for $1\leq
    i<j\leq n$, define $N_{i,j}\colon \proj[j]{R}\to\nats$ by
    \begin{equation*}
       N_{i,j}(a) = |\set{(\bu,a)\in \proj[{[i] \cup \set{j}}]{R}}|\,.
    \end{equation*}
    Since we have
    \begin{equation*}
       |R| = \sum_{a\in\proj[n]{R}} N_{n-1,n}(a)\,,
    \end{equation*}
    we need to compute the function $N_{n-1,n}$, which we do
    iteratively.  For each $j\in \set{2,\dots, n}$, $N_{1,j}(a) = |\set{b\in
    \proj[1]{R}: (b,a)\in \proj[1,j]{R}}|$.  By Lemma~\ref{lem60},
    these quantities can be computed by using $F$ to determine
    $\proj[1,j]{R}$, in total time $\bigO(n^2)$.  (Note, in particular,
    that $|\proj[1,j]{R}| \leq q^2 = \bigO(1)$ and $F$ may be assumed to
    be small so $|F|\leq \bigO(n)$.)  To continue the iteration, we use
    $N_{i,i+1}$ and $N_{i,j}$ to compute $N_{i+1,j}$ for $j = i+2,
    \dots, n$.  We repeat these computations for each $i=1,\dots,n-1$.

    Consider a particular $i$ and $j$ and suppose that we have
    computed $N_{i-1, k}$ for all $k\geq i$.  Let $J=[i]\cup\set{j}$
    and let $H=\proj[J]{R}$, which we will express as a ternary
    relation
    \begin{equation*}
       H = \set{(\bu,x,y)\in\proj[J]{R}:
                                \bu\in\proj[{[i-1]}]{R},\,
                                x\in\proj{R},\, y\in\proj[j]{R}}\,.
    \end{equation*}
    Since $R$ is strongly balanced, the matrix
    \begin{equation*}
       M(x,y) = |\set{\bu\in\proj[{[i-1]}]{R}: (\bu,x,y)\in H}|
    \end{equation*}
    is a rank-one block matrix.  The block structure of $M$ is given
    by the relation \proj[i,j]{R}, since if $(x,y)\in\proj[i,j]{R}$,
    there is at least one $\bt\in R$ with $\proj{\bt}=x$ and $\proj[j]{\bt}
    = y$.  By Lemma~\ref{lem60}, we can compute \proj[i,j]{R} in
    $\bigO(n)$ time, using $F$.

    For notational simplicity, let us write $\cD_i = \proj{R}$.
    Consider $M(\cdot,y)$, the $y$-indexed row of $M$.  We have
    \begin{equation}
    \label{eq25}
        \sum_{x\in \cD_i} M(x,y)
            = \sum_{x\in \cD_i}|\set{\bu: (\bu,x,y)\in H}|
            = |\set{(\bu,x): (\bu,x,y)\in H}|
            = N_{i,j}(y)\,.
    \end{equation}
    Now observe that the relation $B_y(\bu,x) = \set{(\bu,x):
    (\bu,x,y)\in H}$ is rectangular, by Lemma~\ref{lem35}.
    Write $S_y(x) = \set{\bu: (\bu,x,y)\in H}$.  By
    Corollary~\ref{cor10}, there is an equivalence relation on $\cD_j$
    \begin{equation*}
       \theta_y(x_1, x_2) = \exists\bu\,\big(H(\bu,x_1,y)\wedge H(\bu,x_2,y)\big)
    \end{equation*}
    such that $S_y(x_1)$ and $S_y(x_2)$ are equal, if $\theta_y(x_1,
    x_2)$, and disjoint, otherwise.  Thus, if $\cS(y) \subseteq \cD_i$
    contains one representative of each equivalence class of
    $\theta_y$, then
    \begin{equation}
    \label{eq30}
        \sum_{x\in \cS(y)} M(x,y)
            = |\set{\bu: \exists x\,(\bu,x,y)\in H}|
            = N_{i-1,j}(y)\,.
    \end{equation}

    Now, suppose that $\theta_y(x_1,x_2)$ and $y'\neq y$.  Thus,
    $H(\bu, x_1, y)$ and $H(\bu, x_2, y)$ for some $\bu$, so $(x_1,y),
    (x_2,y)\in C$ for some block $C$ of \proj[i,j]{R}.  There is
    $\bu'$ such that $H(\bu'\!,x_1,y')$ if, and only if, $(x_1,y')\in
    C$.  But then we have
    \begin{equation*}
       \begin{array}{c@{\hspace{1cm}}c@{\hspace{1cm}}c}
            \bu'       & x_1 & y'\phantom{,}   \\
            \bu\phdash & x_1 & y\phantom{{}',} \\
            \bu\phdash & x_2 & y\phantom{{}',} \\ \hline
            \bu'       & x_2 & y',
        \end{array}
    \end{equation*}
    and, hence, $\theta_{y'}(x_1, x_2)$.  Thus, the equivalence
    relations $\theta_y$ depend only on the block $C$ containing $y$.
    Thus, we may deduce the classes of $\theta_y$ from \proj[i,j]{R}
    and those of the relation \tw[i,j], defined by
    \begin{equation*}
       x_1\tw[i,j]x_2 \quad \iff \quad
            \exists\bu,y\, \big(H(\bu,x_1,y)\wedge H(\bu,x_2,y)\big)\,.
    \end{equation*}
    We prove in Section~\ref{sec:Dichotomy:Cong}, below, that the
    \tw[i,j] are congruences in \cocl[R].  Thus, the matrix $M$ has
    identical columns corresponding to the equivalence classes of
    \tw[i,j].

    Similarly, there are identical rows corresponding to the
    equivalence classes of \tw[j,i], where
    \begin{equation*}
       y_1\tw[j,i]y_2 \quad \iff \quad
            \exists\bu,x\, \big(H(\bu,x,y_1)\wedge H(\bu,x,y_2)\big)\,.
    \end{equation*}
    (There is no ambiguity of notation between $\tw[i,j]$
    and $\tw[j,i]$ since we have $i < j$.)

    We prove in Section~\ref{sec:Dichotomy:Cong} that the \tw[j,i]
    are also congruences in \cocl[R].  Now, if $\cS'(x)$ contains one
    representative of each of the classes of the corresponding
    equivalence relation $\theta'_x$, we have
    \begin{equation}
    \label{eq40}
        \sum_{y\in \cS'(x)} M(x,y)
            = |\set{\bu: \exists y\,(\bu,x,y)\in H}|
            = N_{i-1,i}(x)\,.
    \end{equation}

    The matrix $\widehat{M}$, obtained by choosing one
    representative from each of the equivalence classes of \tw[i,j]
    and \tw[j,i], is also a rank-one block matrix.  Moreover, we know
    the block structure, row and column sums of $\widehat{M}$, from
    \proj[i,j]{R}, \tw[i,j], \tw[j,i], \eqref{eq30} and \eqref{eq40}.
    Hence, by Lemma~\ref{lem173}, we can reconstruct all the entries
    of $\widehat{M}$.  Then, using \proj[i,j]{R}, \tw[i,j] and
    \tw[j,i], we can reconstruct the matrix $M$.  Finally we compute
    the row sums, as in \eqref{eq25}, to give the value of $N_{i,j}(a)$
    for each $a\in\proj[j]{R}$.

    The time complexity of the algorithm is $\bigO(n)$ for a
    given $i$ and $j$, even in the bit-complexity model.  Since there
    are $\bigO(n^2)$ pairs $i,j$, the overall complexity is $\bigO(n^3)$.

    To complete the proof, we must show how to compute the congruences
    \tw[i,j] and \tw[j,i] in $\bigO(n^5)$ time.  We do this in the following
    section.\qquad
\end{proof}

The time complexity of this algorithm is $\bigO(n^5)$.  However, observe
that the time needed to compute $F$ is already $\bigO(mn^4)$.  We may
assume that $m=\Omega(n)$ as, otherwise, there is a variable, $x_1$
say, which appears in no constraint.  Thus, $x_1$ can be removed to
give a relation $R_1(x_2,\dots,x_n)$ such that $|R| = q|R_1|$.  Therefore, the time complexity of the counting algorithm is no worse than the
$\bigO(mn^4)$ cost of computing the frame $F$.


\subsection{The congruences \tw[i,j] and \tw[j,i]}
\label{sec:Dichotomy:Cong}

We now prove that the relations \tw[i,j] and \tw[j,i] used in the
proof of Lemma~\ref{lem195} are congruences and that they can be
computed efficiently.  Let $\Gamma$ be strongly rectangular and let
$R$ be an $n$-ary relation determined by a $\Gamma$-formula $\Phi$.
For $1 < i < j \leq n$, recall that
\begin{enumerate}[label=(\roman*)]
\item $a \tw[i,j] b$ ($a,b\in\proj[j]{R}$) if there are $\bt,\bt'\in
    R$ such that $\proj[{[i]}]\bt = \proj[{[i]}]\bt'\!$, $t_j=a$ and
    $t'_j=b$;
\item $a \tw[j,i] b$ ($a,b\in\proj{R}$) if there are $\bt,\bt'\in R$
    such that $\proj[J]\bt = \proj[J]\bt'\!$, $t_i=a$ and
    $t'_i=b$,\\ where $J=[i-1]\cup\set{j}$.
\end{enumerate}

\begin{lemma}
\label{lem110}
    For all $1 < i < j \leq n$, \tw[i,j] and \tw[j,i] are congruences
    in \cocl[R].
\end{lemma}
\begin{proof}
    Consider the binary relation $B$ defined by $B(\bu,y) = \exists \bz_1, \bz_2\, R(\bu,\bz_1,y,\bz_2)$ on $\proj[{[i]}]{R} \times
    \proj[j]{R}$.  This is rectangular so induces a
    congruence $\theta_2$ on \proj[j]{R} by Corollary~\ref{cor15}.
    This congruence is \tw[i,j].

    The proof for \tw[j,i] is similar, using $B$ defined by $B(\bu,y)=\exists \bz_1,\bz_2\,     R(\bx,y,\bz_1,w,\bz_2)$ on $\proj[J]{R}     \times \proj{R}$, where $\bu=(\bx,w)$.\qquad
\end{proof}

\begin{lemma}
\label{lem200}
    The set of congruences \tw[i,j] and \tw[j,i] ($1<i<j\leq n$) can be
    computed in $\bigO(n^5)$ time.
\end{lemma}
\begin{proof}
    We compute the relations \tw[i,j], with $i<j$, as follows.  From
    the frame $F$, we compute $\proj[i,j]{R}$.  For each
    $b\in\proj{R}$, this gives a tuple $\bt$ such that $\proj[j]{\bt}
    = b$.  We now use Corollary~\ref{cor100}, to compute a frame $F^\star$
    for $R(t_1,\dots,t_i,x_{i+1},\dots,x_n)$ in $\bigO(n^3)$ time.  Now
    $\proj[j]{F^\star}$ gives the equivalence class of \tw[i,j]
    containing $b$.  We repeat this procedure, as in the proof of
    Lemma~\ref{lem140}, until we have determined all the equivalence
    classes.

    There are $\bigO(n^2)$ pairs $i,j$ with $i<j$ and computing each
    \tw[i,j] requires $\bigO(n^3)$ time.  Thus, the we can compute all
    \tw[i,j] in $\bigO(n^5)$ time.

    Now consider the relations \tw[j,i], with $i<j$.  For each
    $a\in\proj{R}$, compute a frame $F_{j,a}$ for the relation
    $R_{j,a}$ determined by $\Phi\wedge \chi_a(x_j)$.  (Recall that
    $\chi_a$ is the relation containing only $a$ and we may assume that
    $\chi_a\in \Gamma$ by Lemma~\ref{lem35}.)  From
    Lemma~\ref{lem140}, we can do this in $\bigO(n^4)$ time, so $\bigO(n^5)$
    time in total.  Now, for each $i<j$, determine $\proj[i,j]{R}$,
    using $F$.  This requires $\bigO(n)$ time for each pair $i,j$, so
    $\bigO(n^3)$ time in total.

    Now, for each block $C$ of \proj[i,j]{R}, choose $a\in\proj[j]{R}$
    so that $(x,a)\in C$ for some $x\in\proj{R}$.  Then the congruence
    \tw[i] of $R_{j,a}$ gives the equivalence classes of \tw[j,i]
    corresponding to $C$.  These can be determined in $\bigO(n)$ time
    using $F_{i,a}$.  Thus, the total time to compute \tw[j,i] for all
    pairs $i,j$ with $i<j$ is $\bigO(n^5)$.

    Hence the total time needed to compute all of these congruences is
    $\bigO(n^5)$.\qquad
\end{proof}


\section{Decidability}
\label{sec:Decide}
Having shown that \numcsp has a dichotomy, we must consider whether it
is effective. That is, given a relational structure $\frS=(D,\Gamma)$
can we decide algorithmically whether the problem \ncsp is in \fp or
is \numpc? This is the major question left open
in~\cite{Bulato08}. Here we show that the answer is in the
affirmative.

We will construct an algorithm to solve the following decision problem.\vspace{1ex}

\begin{tabular}{lll@{\,}cl}
\quad&\multicolumn{4}{l}{\strbal}\\[0.25ex]
&\ \ &\textsf{Instance}&:& A relational structure $\frS=(D,\Gamma)$. \\
&\ \ &\textsf{Question}&:& Is $\Gamma$ strongly balanced?
\end{tabular}\vspace{1ex}

Recall from Section~\ref{sec:Defs} that we may assume that
$\size{\Gamma}\geq q$.  Thus, we may take $\size{\Gamma}$ as the
measure of input size for \strbal. We bound the complexity of \strbal
as a function of $\size{\Gamma}$.  Complexity is a secondary issue,
since $\size{\Gamma}$ is a constant in the nonuniform model for
\ncsp. In the nonuniform model, we are only required to show that some
algorithm exists to solve \strbal. However, we believe that the
computational complexity of deciding the dichotomy is intrinsically
interesting.

Our approach will be to show that the strong balance condition is equivalent to a structural property of $\Gamma$ that can be checked in \np.

We must first verify that $\Gamma$ is strongly rectangular, since otherwise it cannot be strongly balanced, by Lemma~\ref{lem175}. Thus, we consider the following computational problem.\vspace{1ex}

\begin{tabular}{lll@{\,}cl}
\quad&\multicolumn{4}{l}{\strect}\\[0.25ex]
&\ \ &\textsf{Instance}&:& A relational structure $\frS=(D,\Gamma)$. \\
&\ \ &\textsf{Question}&:& Is $\Gamma$ strongly rectangular?
\end{tabular}\vspace{1ex}

\begin{lemma}\label{lem280}
\strect is in \np.
\end{lemma}
\begin{proof}
By Lemma~\ref{lem30}, we can verify that a given function \ph is a \maltsev polymorphism in $\bigO(\size{\Gamma}^4)$ time. Thus, we select a function $\ph\colon D^3\to D$ nondeterministically in $\bigO(q^3)=\bigO(\size{\Gamma}^3)$ time and check that it is a \maltsev polymorphism in a further $\bigO(\size{\Gamma}^4)$ time.\qquad
\end{proof}

The remainder of this section is organised as follows.  We first give
definitions and notation that were held over from
Section~\ref{sec:Defs} because they are only used here.  In
Section~\ref{sec:Dec:Rank-one}, we give a characterisation of rank-one
block matrices that we use in our decidability proof.  The proof
itself appears in Section~\ref{sec:Dec:Proof}.


\subsection{Definitions and notation}

An equivalent but different view of \cspg from the one we
have used is often taken in the literature. This is to regard $\Phi$
as a finite structure with domain $V$ and relations determined by the
scopes of the constraints. Thus, we have relations $\tilde{H}$, where
$(i_1, i_2, \dots, i_r)\in \tilde{H}$ if $H(x_{i_1}, x_{i_2}, \dots,
x_{i_r})$ is a constraint.  In this view, a satisfying assignment $\bx$ is a
\emph{homomorphism} from $\Phi$ to $\Gamma\!$.

The following definitions and notation will be used in the remainder
of this section.  Let $[D_1\to D_2]$ denote the set of functions from
$D_1$ to $D_2$. Then a homomorphism between two relational structures
$\frS_1=(D_1,\Gamma_{\!1})$, $\frS_2=(D_2,\Gamma_{\!2})$ is a function
$\sigma\in[D_1\to D_2]$ that preserves relations. Thus, for each
$r$-ary relation $H_1\in\Gamma_{\!1}$ there is a corresponding $r$-ary
relation $H_2\in\Gamma_{\!2}$ and, for each tuple $\bu = (u_1, \ldots,
u_r)\in H_1$, we have $\sigma(\bu) = (\sigma(u_1), \ldots,
\sigma(u_r)) \in H_2$. We will write $\sigma\colon\frS_1\to\frS_2$ to
indicate that $\sigma$ is a homomorphism.

Let $[V\inj D]$ denote the set of all \emph{injective} functions $V\to
D$ and let $[V\bij D]$ denote the set of all \emph{bijective}
functions $V\to D$. If $\sigma\colon\frS_1\to\frS_2$ and
$\sigma\in[D_1\inj D_2]$, then $\sigma$ is called a
\emph{monomorphism} and we will write $\sigma\colon\frS_1\inj \frS_2$.
If $\sigma$ is a bijective homomorphism and $\sigma^{-1}$ is also a
homomorphism, then $\sigma$ is called an \emph{isomorphism} and we
write $\sigma\colon \frS_1\bij \frS_2$.  Then $\frS_1$, $\frS_2$ are
\emph{isomorphic}, so isomorphic structures are the same up to
relabelling. An \emph{endomorphism} of a relational structure $\frS$
is a homomorphism $\sigma\colon\frS\to\frS$ and an \emph{automorphism}
is an isomorphism $\sigma\colon\frS\bij\frS$. Note that the definition
of an endomorphism is identical to that of a unary polymorphism. Note
also that $[D\inj D]=[D\bij D]$, since $D$ is finite, so an injective
endomorphism is always an automorphism. Clearly, the identity function
is always an automorphism, for any relational structure $\frS$.

We use the following construction of \emph{powers of $\frS$} (see, for example,~\cite[p.\,282]{NeSiZa10}). For any relational structure $\frS=(D,\Gamma)$ and $k\in\nats$, the relational structure $\frS^k=(D^k\!,\Gamma^k)$ is defined as follows. The domain is the Cartesian power $D^k\!$. The constraint language $\Gamma^k$ is such that, for each $r$-ary relation $H\in\Gamma\!$, there is an $r$-ary $H^k\in\Gamma^k\!$, which is defined to be the following relation. If $\bu_i=(u_{i,1},u_{i,2},\ldots,u_{i,k}) \in D^k$ ($i\in[r]$), then $(\bu_1,\bu_2,\ldots,\bu_r) \in H^k$ if, and only if,
$(u_{1,j},u_{2,j},\ldots,u_{r,j}) \in H$ for all $j\in[k]$.  Now, if
$\Psi$ is a pp-formula in $\Gamma\!$, we define the corresponding
formula $\Psi^k$ to be identical to $\Psi$, except that each
occurrence of $H\in\Gamma$ is replaced by the corresponding relation
$H^k\in\Gamma^k\!$. Observe that the relation $\Psi^k$ is actually
pp-definable in $\Gamma\!$, by the formula $\Psi^k(\bx) = \Psi(\bx_1) \wedge \Psi(\bx_2) \wedge\cdots\wedge\Psi(\bx_k)$,
where $\bx_i$ ($i\in[k]$) are disjoint $n$-tuples of variables. In particular, we have $|\Psi^k|=|\Psi|^k\!$.

Using this construction, the definition of a polymorphism can be reformulated. In this view of \cspg, it follows directly that a $k$-ary polymorphism is just a homomorphism $\psi\colon\frS^k\to\frS$.


\subsection{Rank-one block matrices}
\label{sec:Dec:Rank-one}

In our decidability proof, we use a different characterisation of
rank-one block matrices, given by Corollary~\ref{cor110}. This may
seem more complicated than the original definition but it is more
suited to our purpose.

\begin{lemma}\label{lem210}
A matrix $A$ is a rank-one block matrix if, and only if, every $2\times 2$ submatrix of $A$ is a rank-one block matrix.
\end{lemma}
\begin{proof}
Let $A$ be a $k\times\ell$ rank-one block matrix and let
\begin{equation*}
    B\ =\ \begin{bmatrix} a_{ir} & a_{is}\\a_{jr} & a_{js}\end{bmatrix}
        \qquad \text{($i,j\in[k]$, $i\neq j$; $r,s\in[\ell]$, $r\neq s$)}.
\end{equation*}
be any $2\times 2$ submatrix of $A$. If any of $a_{ir}$, $a_{is}$, $a_{jr}$, $a_{js}$ is zero, at least two must be zero, since $A$ is rectangular. Then $B$ is clearly a rank-one block matrix. If $a_{ir}$, $a_{is}$, $a_{jr}$, $a_{js}$ are all nonzero, $B$ must be a submatrix of some block of $A$. Since this block has rank one, $B$ also has rank one.

Conversely, suppose $A$ is not a rank-one block matrix. If its underlying relation is not rectangular, there exist $a_{ir},a_{is},a_{jr}>0$ with
$a_{js}=0$. The corresponding matrix $B$ clearly has rank $2$, but has only one block so is not a rank-one block matrix. If the underlying relation of $A$ is rectangular, then $A$ must have a block of rank at least~$2$. This block must have some $2\times2$ submatrix $B$ with rank $2$ and all its elements $a_{ir},a_{is},a_{jr},a_{js}>0$.\qquad
\end{proof}
\begin{lemma}\label{lem220}
    Let $A$ be a rectangular $2\times 2$ matrix. $A$ is a rank-one block matrix
    if, and only if, $a_{11}^2a_{22}^2a_{12}a_{21} = a_{12}^2a_{21}^2a_{11}a_{22}$.
\end{lemma}
\begin{proof}
This equation holds if any of $a_{11}$, $a_{22}$, $a_{12}$ or $a_{21}$ is zero. But then rectangularity implies that at least two of them must be zero and $A$ is a  rank-one block matrix in all possible cases.
Otherwise, the equation is equivalent to $a_{11}a_{22}=a_{12}a_{21}$, which is the condition that $A$ is singular. So $A$ is one block, with rank one. The argument is clearly reversible.\qquad
\end{proof}
\begin{corollary}\label{cor110}
    Ket $A$ be a rectangular $k\times\ell$ matrix. $A$ is a rank-one
    block matrix if, and only if, $a_{ir}^2a_{js}^2a_{is}a_{jr} =
    a_{is}^2a_{jr}^2a_{ir}a_{js}$ for all $i,j\in[k]$ and all
    $r,s\in[\ell]$.
\end{corollary}
\begin{proof}
When $i=j$ or $r=s$, the two sides of this equation are identical. Otherwise, the equality follows directly from Lemmas~\ref{lem210} and~\ref{lem220}.\qquad
\end{proof}
\begin{remark}\label{rem55}%
    It is possible to modify the above so that Corollary~\ref{cor110} involves products of only five elements, rather than six, but we do not pursue this refinement.
\end{remark}


\subsection{Decidability}
\label{sec:Dec:Proof}

To show the decidability of strong balance, we relax the criterion of
strong balance, by noting the conditions sufficient for the success of
the algorithm in Section~\ref{sec:Dichotomy:Alg}. Observe that only
ternary relations on $D\times D\times D^{i}\!$, for $i\in[n-2]$, are
required to be balanced.  Therefore, let $\Psi(\bx)$, with $\bx=(x_1,\ldots,x_n)$, be an arbitrary formula pp-definable in $\Gamma\!$, which we consider fixed for the rest of this section. Then, for the algorithm to succeed, it suffices that the
$q\times q$ matrix
\begin{equation*}
 M(a,b)\, = \, \big|\set{\bx\in [V\to D] : \bx\in\Psi,\,x_1=a,\,x_2=b}\big|\qquad(\forall a,b\in D)
\end{equation*}
is always a rank-one block matrix. Note that we can always assume that the underlying relation of $M$ is rectangular, since $\Gamma$ is known to be strongly rectangular.
\begin{remark}\label{fpnumprem}%
    Call this condition \emph{almost-strong} balance. It is equivalent
    to strong balance if $\fp\neq\nump$. If \frS is strongly balanced,
    it is clearly almost-strongly balanced. Almost-strong balance
    implies that the algorithm of Section~\ref{sec:Dichotomy:Alg}
    succeeds, which implies that $\ncsp\in\fp$. Thus \ncsp is not
    \numpc, which implies that it is strongly balanced by
    Lemma~\ref{lem180}. This chain of implications requires
    $\fp\neq\nump$, so we make that assumption in the remainder of
    this section. If $\fp=\nump$, no dichotomy exists and the property
    of strong balance ceases to be of computational interest.
\end{remark}

We may therefore take almost-strong balance as the criterion for strong balance. By Corollary~\ref{cor110}, the condition for $M$ to be a rank-one block matrix is that
\begin{equation}\label{eq:M}
M(a,c)^2M(a,d)M(b,d)^2M(b,c)\ =\ M(a,d)^2M(a,c)M(b,c)^2M(b,d),
\end{equation}
for all $a,b,c,d\in D$.

We can reformulate the condition for strong balance using the construction of powers of \frS.
If $\ba=(a_1,\ldots,a_k)$ and $\bb=(b_1,\ldots,b_k)$, the balance matrix $M_k$ for $\Psi^k$ is the $q^k\times q^k$ matrix
\begin{align*}
M_k(\ba,\bb)\,&=\,\big|\set{\bx\in [V\to D^k]\,:\,\bx\in\Psi^k,\,x_1=\ba,x_2=\bb}\big|\\
&=\,M(a_1,b_1)M(a_2,b_2)\cdots M(a_k,b_k)\,.
\end{align*}
Using this, equation~\eqref{eq:M} can be rewritten as
\begin{equation}\label{eq:M6}
M_6(\eua,\euc)\,=\,M_6(\eua,\eud)\,,
\end{equation}
where
\begin{equation}\label{eq:abcd}
\eua=(a,a,a,b,b,b), \ \euc=(c,c,d,d,d,c),\ \eud=(d,d,c,c,c,d)\,.
\end{equation}
Fix $\eua$, $\euc$, $\eud$ and, for notational simplicity,
write $\euS$ for $\frS^6\!$, $\euGam$ for $\Gamma^6\!$, $\euPsi$ for
$\Psi^6\!$, $\euM$ for $M_6$ and $\euD$ for $D^6\!$. Then, from
\eqref{eq:M6}, we must verify that $\euM(\eua,\euc)=\euM(\eua,\eud)$
for all relations $\euPsi$ which are pp-definable in $\euGam$ and
given $\eua,\euc,\eud\in \euD$. We use a method of
Lov\'{a}sz~\cite{Lovasz67}; see also~\cite{DyGoPa07}. For $\eus\in \euD$, let
\begin{align*}
    \Hom_{\eua,\eus}(\euPsi) &= \set{\bx\in [V\to \euD]:
                               \bx\in \euPsi,\, x_1=\eua,\, x_2=\eus} \\
    \hom_{\eua,\eus}(\euPsi) &= |\Hom_{\eua,\eus}(\euPsi)|\,.
\end{align*}

However, a homomorphism $V\to \euD$ that is consistent with $\euPsi$ is
just a satisfying assignment to $\euPsi$.  $\euM(\eua,\eus)$ is the
number of such assignments with $x_1=\eua$ and $x_2=\eus$, i.e., the
number of homomorphisms that map $x_1\mapsto \eua$ and $x_2\mapsto
\eus$.  This proves the following.

\begin{lemma}
\label{lem229}
    $\Gamma$ is strongly balanced if, and only if,
    $\hom_{\eua,\euc}(\euPsi)=\hom_{\eua,\eud}(\euPsi)$ for all formulae $\euPsi$
    and all $\eua, \euc, \eud$ of the form above.
\end{lemma}

We will also need to consider the injective functions in
$\Hom_{\eua,\eus}(\euPsi)$.  For $\eus\in \euD$, let
\begin{align*}
    \Mon_{\eua,\eus}(\euPsi) &= \set{\bx\in [V\inj\euD]:
                               \bx\in \euPsi,\, x_1=\eua,\, x_2=\eus} \\
    \mon_{\eua,\eus}(\euPsi) &= |\Mon_{\eua,\eus}(\euPsi)|\,.
\end{align*}

\begin{lemma}\label{lem230}
    $\hom_{\eua,\euc}(\euPsi)=\hom_{\eua,\eud}(\euPsi)$ for all $\euPsi$ if, and only if, $\mon_{\eua,\euc}(\euPsi)=\mon_{\eua,\eud}(\euPsi)$ for all~$\euPsi$.
\end{lemma}
\begin{proof}
Consider the set $\cI$ of all partitions $I$ of $V$ into disjoint classes $\euI_1,\ldots,\euI_{k_I}$, such that $1\in \euI_1$, $2\in \euI_2$. Writing $I \preceq I'$  whenever $I$ is a refinement of $I'\!$, $\bP=(\cI,\preceq)$ is a poset. We will write $\bot$ for the partition into singletons, so $\bot\preceq I$ for all $I\in \mathcal{I}$.

Let $V/I$ denote the set of classes $\euI_1,\ldots,\euI_{k_I}$ of the
partition $I$, so $|V/I|=k_I$, and let $\euI_1$, $\euI_2$ be denoted
by $1/I$, $2/I$. Let $\euPsi/I$ denote the relation obtained from
$\euPsi$ by imposing equality on all pairs of variables that occur in
the same partition of $I$.  Thus, the constraints $x_1=\eua$, $x_2=\eus$ become $x_{1/I}=\eua$, $x_{2/I}=\eus$. Then we have
\begin{equation}
\label{eq:poset}
    \hom_{\eua,\eus}(\euPsi)
      \ =\ \hom_{\eua,\eus}(\euPsi/\bot)
      \ = \sum_{I\in\mathcal{I}} \mon_{\eua,\eus}(\euPsi/I)
      \ = \sum_{I\in\mathcal{I}} \mon_{\eua,\eus}(\euPsi/I)\zeta(\bot,I)\,,
\end{equation}
where $\zeta(I,I')=1$, if $I \preceq I'\!$, and $\zeta(I,I')=0$,
otherwise, is the $\zeta$-function of the poset $\bP$. Thus, if
$\mon_{\eua,\euc}(\euPsi)=\mon_{\eua,\eud}(\euPsi)$ for all $\euPsi$, then
\begin{equation}
\label{eq:hom}
    \hom_{\eua,\euc}(\euPsi)
      \ = \sum_{I\in\cI} \mon_{\eua,\euc}(\euPsi/I)\zeta(\bot,I)
      \ = \sum_{I\in\mathcal{I}} \mon_{\eua,\eud}(\euPsi/I)\zeta(\bot,I)
      \ =\ \hom_{\eua,\eud}(\euPsi)\,.
\end{equation}

More generally, the reasoning used to give \eqref{eq:poset} implies that
\begin{equation*}
    \hom_{\eua,\eus}(\euPsi/I)
      \ = \sum_{I\preceq I'} \mon_{\eua,\eus}(\euPsi/I')
      \ = \sum_{I'\in\mathcal{I}} \mon_{\eua,\eus}(\euPsi/I')\zeta(I,I')\,.
\end{equation*}
Now, M\"obius inversion for posets~\cite[Ch.\,25]{VanLint01} implies
that the matrix $\zeta\colon \mathcal{I}\times \mathcal{I}\rightarrow
\{0,1\}$ has an inverse $\mu\colon \mathcal{I}\times
\mathcal{I}\rightarrow \ints$. It follows directly that
\begin{equation*}
    \mon_{\eua,\eus}(\euPsi)\ =\ \sum_{I\in\cI} \hom_{\eua,\eus}(\euPsi/I)\mu(\bot,I)\,.
\end{equation*}

Thus, if $\hom_{\eua,\euc}(\euPsi)=\hom_{\eua,\eud}(\euPsi)$ for all $\euPsi$, then
\begin{equation}
\label{eq:mon}
    \mon_{\eua,\euc}(\euPsi)
      \ = \sum_{I\in\cI} \hom_{\eua,\euc}(\euPsi/I)\mu(\bot,I)
      \ = \sum_{I\in\mathcal{I}} \hom_{\eua,\eud}(\euPsi/I)\mu(\bot,I)
      \ =\ \mon_{\eua,\eud}(\euPsi)\,.
\end{equation}
Now,~\eqref{eq:hom} and~\eqref{eq:mon} give the conclusion.\qquad
\end{proof}

\begin{lemma}\label{lem260}
$\mon_{\eua,\euc}(\euPsi)=\mon_{\eua,\eud}(\euPsi)$,  for all $\euPsi\!$, if, and only if, there is an automorphism $\eta\colon\euD\bij \euD$ of $\euS=(\euD,\euGam)$ such that $\eta(\eua)=\eua$ and $\eta(\euc)=\eud$.
\end{lemma}
\begin{proof}
The condition holds if $\euS$ has such an automorphism since, if
$\euPsi(\bx)=\exists\by\,\euPhi(\bx,\by)$ for some $\euPhi$, then
\begin{align*}
    \mon_{\eua,\euc}(\euPsi)\ &=\ |\set{\bx\in[V\inj\euD] : x_1=\eua,\,x_2=\euc,\, \exists\by\,(\bx,\by)\in\euPhi}| \\
        &=\ |\set{\eta(\bx)\in[V\inj\euD] : x_1=\eta(\eua),\,x_2=\eta(\euc),\, \exists\by\,(\eta(\bx),\eta(\by))\in\euPhi}|\\
        &=\ |\set{\bx\in[V\inj\euD]:x_1=\eua,\,x_2=\eud,\, \exists\by\,(\bx,\by)\in\euPhi}|\\
        &=\ \mon_{\eua,\eud}(\euPsi)\,.
\end{align*}

For the converse, suppose we have $\mon_{\eua,\euc}(\euPsi)=\mon_{\eua,\eud}(\euPsi)$ for all $\euPsi$. Consider the following $\euGam$-formula $\euPhi$ with domain $\euD$ and variables $x_\eui$ ($\eui\in\euD$),
\begin{equation*}
\euPhi(\bx)\ =\ \bigwedge_{\euH\,\in\,\euGam}\, \bigwedge_{(\euu_1,\ldots,\euu_r)\,\in\,\euH} \euH(x_{\euu_1},\ldots,x_{\euu_r})\,.
\end{equation*}
Then
\begin{equation*}
\Mon_{\eua,\eus}(\euPhi)\ =\ \set{\bx\in[\euD\inj\euD]:  \ x_\eua=\eua, \,x_\euc=\eus,\,\bx\in\euPhi}\,.
\end{equation*}
We have $\Mon_{\eua,\euc}(\euPhi)\neq\emptyset$, since the identity assignment $x_\eui=\eui$ $(\eui\in\euD)$ is clearly satisfying. Thus, by the assumption, $\Mon_{\eua,\eud}(\euPhi)\neq\emptyset$. Let $\eta\in\Mon_{\eua,\eud}(\euPhi)$, so $\eta$ is an endomorphism of $\euS$ with $\eta(\eua)=\eua$, $\eta(\euc)=\eud$. Since $[D\inj D]=[D\bij D]$, $\eta\colon D\bij D$ is the required automorphism.\qquad
\end{proof}

\begin{corollary}\label{cor140}
$\frS=(D,\Gamma)$ is strongly balanced if, and only if, for all $a,b,c,d\in D$ and $\eua,\euc,\eud$ as defined in \eqref{eq:abcd}, $\euS=(\euD,\euGam)$ has an automorphism $\eta$ such that $\eta(\eua)=\eua$ and $\eta(\euc)=\eud$.
\end{corollary}
\begin{proof}
This follows from~\eqref{eq:M6} and Lemmas~\ref{lem229}, \ref{lem230} and~\ref{lem260}.\qquad
\end{proof}

This characterisation of strong balance leads to a nondeterministic algorithm.

\begin{theorem}\label{thm60}
\strbal is in \np.
\end{theorem}
\begin{proof}
    We first determine whether $\Gamma$ is strongly rectangular, using
    the method of Lemma~\ref{lem280}.  If it is not, then $\Gamma$ is
    not strongly rectangular by Lemma~\ref{lem175}.

    Otherwise, we can construct $\euS=(\euD,\euGam)$ in time
    $\bigO(\size{\Gamma}^6)$. Let $\euq=q^6=|\euD|$ and let $\Pi$ denote
    the set of $\euq!$ permutations of $\euD$.  Each $\pi\in\Pi$ is a
    function $\pi\colon\euD\inj\euD$ and so a potential automorphism
    of $\euS$. For each of the $q^4$ possible choices $a,b,c,d\in D$,
    we determine $\eua,\euc,\eud\in\euD$ in polynomial time. We select
    $\pi\in\Pi$ nondeterministically and check that $\pi(\eua)=\eua$,
    $\pi(\euc)=\eud$ and that $\pi$ preserves all $\euH\in\euGam$. The
    computation requires $\bigO(q^4\size{\euGam}^2)=\bigO(\size{\Gamma}^{16})$
    time in total, so everything other than the $\bigO(q^{10}) =
    \bigO(\size{\Gamma}^{10})$ nondeterministic choices can be done
    deterministically in a polynomial number of steps.\qquad
\end{proof}

\begin{remark}\label{rem65}%
    We have paid little attention to the efficiency of the
computations in Theorem~\ref{thm60}. If the elements of $D$ are encoded as binary numbers in $[q]$, comparisons and nondeterministic choices require $\bigO(\log q)$ bit operations, rather than the $\bigO(1)$ operations in our accounting. On the other hand, membership in $H^6$ can be tested in $\bigO(\size{H})$ comparisons, rather than the $\bigO(\size{H}^6)$ that we have allowed. This might be reduced further by storing $H$ in a suitable data structure, instead of a simple matrix. We could also use Remark~\ref{rem55} to improve the algorithm of Theorem~\ref{thm60}.
\end{remark}
\begin{remark}\label{rem70}%
    Theorem~\ref{thm60} and Lemma~\ref{lem190} together imply that the following problem, posed by Bulatov~\cite{Bulato08}, can also be decided in \np.\vspace{1ex}

\emph{\begin{tabular}{lll@{\,}cl}
\quad&\multicolumn{4}{l}{\congsing}\\[0.25ex]
&\ \ &\textsf{Instance}&:& A relational structure $\frS=(D,\Gamma)$. \\
&\ \ &\textsf{Question}&:& Is $\Gamma$ congruence singular?
\end{tabular}}\vspace{1ex}

  Whether this can be shown directly, and not via \strbal, remains open.
\end{remark}


\section{Conclusions}
\label{sec:Conclude}
We have shown that there is an effective dichotomy for the whole of \numcsp. We have given a new, and simpler, proof for the existence of the dichotomy and the first proof of its decidability.

The complexity of our counting algorithm is $\bigO(n^5)$, whereas algorithms for most known counting dichotomies are of lower complexity, often $\bigO(n)$.
Can the complexity of the general algorithm be improved to
$\bigO(n^4)$, or better? Since frames, on which the algorithm is
based, have size $\bigO(n)$, there is no obvious reason why this
cannot be done.

A second problem that we have not yet considered is an extension to a dichotomy for \emph{weighted} counting problems~\cite{BulGro05,DyGoJe09}. We believe that this is possible. In fact, a dichotomy for \emph{rational} weights has already been shown in~\cite{BDGJJR10}. This gives an indirect argument, using the unweighted dichotomy. Decidability of the dichotomy of~\cite{BDGJJR10} now follows from Section~\ref{sec:Decide} of this paper.

A third issue is to investigate whether known counting dichotomies can be recovered from these general theorems. We have some preliminary results in this direction. The characterisation of Lemma~\ref{lem260} appears to be useful in this respect.

A fourth problem is to determine the complexity of \strbal more
precisely, rather than just establishing membership in \np.  \strbal
seems unlikely to be \npc as the automorphism tests required can be
coded into a single instance of the graph isomorphism problem.
However, it is not obvious whether the converse reduction is possible
so it may be that \strbal is in \ptime.

Finally, a deeper question that arises from our work is to what
extent the detailed properties of the algebras associated with \csp
instances are of real significance. In recent years, the algebraic
approach has proven successful in the study of \csp, but it is
possible that these algebras are more complicated objects than the
relations they are intended to capture.

\emph{Note.} Since this paper was written, Cai, Chen and Lu have
extended and strengthened our methods to give an effective dichotomy
for the weighted counting problem \cite{CaChLu10}.

\emph{Acknowledgments.}  The authors are grateful to Jin-Yi Cai,
Xi Chen and Andrei Krokhin for carefully reading drafts of an earlier version of this paper.  We are also grateful to Andrei Bulatov for explaining parts of his proof, and to Leslie Ann Goldberg for useful discussions. We also thank a referee for pointing out the issue discussed in Remark~\ref{fpnumprem}.



\begin{thebibliography}{10}

\bibitem{Bulato06}
A.~A. Bulatov.
\newblock A dichotomy theorem for constraint satisfaction problems on a
  3-element domain.
\newblock {\em Journal of the ACM}, 53(1):66--120, 2006.

\bibitem{Bulato07}
A.~A. Bulatov.
\newblock The complexity of the counting constraint satisfaction problem.
\newblock {\em Electronic Colloquium on Computational Complexity},
  14(093), 2007.
\newblock (Revised Feb. 2009).

\bibitem{Bulato08}
A.~A. Bulatov.
\newblock The complexity of the counting constraint satisfaction problem.
\newblock In {\em Proc. 35th International Colloquium on Automata, Languages
  and Programming (Part 1)}, LNCS  5125, pp. 646--661. Springer, 2008.

\bibitem{BulDal06}
A.~A. Bulatov and V.~Dalmau.
\newblock A simple algorithm for {M}al'tsev constraints.
\newblock {\em SIAM Journal on Computing}, 36(1):16--27, 2006.


\bibitem{BulDal03}
A.~A. Bulatov and V.~Dalmau.
\newblock Towards a dichotomy theorem for the counting constraint satisfaction
  problem.
In {\em Proc. 44th Annual IEEE Symposium on Foundations of Computer Science}, pp.\,562--573, IEEE, 2003.

\bibitem{BulDal07}
A.~A. Bulatov and V.~Dalmau.
\newblock Towards a dichotomy theorem for the counting constraint satisfaction
  problem.
\newblock {\em Information and Computation}, 205(5):651--678, 2007.

\bibitem{BDGJJR10}
A. A. Bulatov, M. E. Dyer, L. A. Goldberg, M. Jalsenius, M. R Jerrum and D. Richerby. The complexity of weighted and unweighted \#CSP. \texttt{arXiv}:1005.2678 [cs.CC], May 2010.

\bibitem{BulGro05}
A.~A. Bulatov and M.~Grohe.
\newblock The complexity of partition functions.
\newblock {\em Theoretical Computer Science}, 348(2--3):148--186, 2005.


\bibitem{CaChLu10}
J.-Y. Cai, X. Chen and P. Lu,
\newblock Non-negative weighted \#CSPs: An effective complexity dichotomy,
\newblock \texttt{arXiv}: 1012.5659 [cs.CC], December 2010.

\bibitem{CaLuXi09}
J.-Y. Cai, P.~Lu, and M.~Xia.
\newblock Holant problems and counting {CSP}.
\newblock In {\em Proc. 41st Annual ACM Symposium on Theory of Computing},
  pp. 715--724. ACM, 2009.

\bibitem{CreHer96}
N.~Creignou and M.~Hermann.
\newblock Complexity of generalized satisfiability counting problems.
\newblock {\em Information and Computation}, 125(1):1--12, 1996.

\bibitem{DenWis02}
K.~Denecke and S.~L. Wismath.
\newblock {\em Universal Algebra and Applications in Theoretical Computer
  Science}.
\newblock Chapman and Hall/CRC, 2002.

\bibitem{DyGoJe08}
M.~E. Dyer, L.~A. Goldberg, and M.~R.~Jerrum.
\newblock A complexity dichotomy for hypergraph partition functions.
\newblock {\em Computational Complexity}, 19(4):605--633, 2010.

\bibitem{DyGoJe09}
M.~E. Dyer, L.~A. Goldberg, and M.~R. Jerrum.
\newblock The complexity of weighted {B}oolean \#CSP.
\newblock {\em SIAM Journal on Computing}, 38(5):1970--1986, 2009.

\bibitem{DyGoPa07}
M.~E. Dyer, L.~A. Goldberg, and M.~S. Paterson.
\newblock On counting homomorphisms to directed acyclic graphs.
\newblock {\em Journal of the ACM}, 54(6), 2007.

\bibitem{DyeGre00}
M.~E. Dyer and C.~S. Greenhill.
\newblock The complexity of counting graph homomorphisms.
\newblock {\em Random Structures and Algorithms}, 17(3--4):260--289, 2000.
\newblock (Corrigendum in \emph{Random Structures and Algorithms},
  25(3):346--352, 2004.).

%

\bibitem{FedVar98}
T.~Feder and M.~Y. Vardi.
\newblock The computational structure of monotone monadic {SNP} and constraint
  satisfaction: {A} study through {D}atalog and group theory.
\newblock {\em SIAM Journal on Computing}, 28(1):57--104, 1998.

\bibitem{FreMck87}
R.~Freese and R.~McKenzie.
\newblock {\em Commutator Theory for Congruence Modular Varieties}.
\newblock Cambridge University Press, 1987.

\bibitem{Geiger68}
D.~Geiger.
\newblock Closed systems of functions and predicates.
\newblock {\em Pacific Journal of Mathematics}, 27:95--100, 1968.

\bibitem{HelNes90}
P.~Hell and J.~Ne\v{s}et\v{r}il.
\newblock On the complexity of ${H}$-coloring.
\newblock {\em Journal of Combinatorial Theory (Series B)}, 48(1):92--110,
  1990.

\bibitem{HobMcK88}
D.~Hobby and R.~McKenzie.
\newblock {\em The Structure of Finite Algebras}, vol.~76 of {\em
  Contemporary Mathematics}.
\newblock American Mathematical Society, 1988.

\bibitem{KolVar98a}
P.~G. Kolaitis and M.~Y. Vardi.
\newblock Conjunctive-query containment and constraint satisfaction.
\newblock In {\em Proc. 17th ACM Symposium on Principles
  of Database Systems (PODS '98)}, pp. 205--213, New York, 1998. ACM.

\bibitem{Ladner75}
R.~E. Ladner.
\newblock On the structure of polynomial time reducibility.
\newblock {\em Journal of the ACM}, 22(1):155--171, 1975.

\bibitem{Lovasz67} L. Lov\'asz. Operations with structures. \emph{Acta. Math. Acad. Sci. Hung.}, 18:321--328, 1967.

\bibitem{NeSiZa10}
J. Ne\v{s}et\v{r}il, M.~H. Siggers and L. Z\'adori.
\newblock A combinatorial constraint satisfaction problem dichotomy
  classification conjecture.
\newblock {\em European Journal of Combinatorics}, 31(1):280--296, 2010.

\bibitem{Schaef78}
T.~Schaefer.
\newblock The complexity of satisfiability problems.
\newblock In {\em Proc. 10th Annual ACM Symposium on Theory of Computing},
  pp. 216--226. ACM Press, 1978.

\bibitem{Toda89}
S.~Toda.
\newblock On the computational power of \PP and \parityP.
\newblock In {\em Proc. 30th Annual Symposium on Foundations of Computer
  Science}, pp. 514--519. IEEE Computer Society, 1989.

\bibitem{Valian79a}
L.~G. Valiant.
\newblock The complexity of computing the permanent.
\newblock {\em Theoretical Computer Science}, 8:189--201, 1979.

\bibitem{Valian79b}
L.~G. Valiant.
\newblock The complexity of enumeration and reliability problems.
\newblock {\em SIAM Journal on Computing}, 8(3):410--421, 1979.

\bibitem{VanLint01} J. van Lint and R. Wilson. \textit{A Course in Combinatorics} (2nd ed.). CUP, 2001.
\end{thebibliography}
\end{document}